\documentclass[12pt, draftclsnofoot, onecolumn]{IEEEtran}

\usepackage[dvips]{color}
\usepackage{comment}
\usepackage{todonotes}
\usepackage{epsf}
\usepackage{epsfig}
\usepackage{times}
\usepackage{epsfig}
\usepackage{graphicx}
\usepackage{bbold}
\usepackage{mathtools}
\usepackage{mathrsfs}
\usepackage{amssymb}
\usepackage{pdfpages}
\usepackage{epstopdf}
\usepackage{dsfont}
\usepackage{lettrine} 
\usepackage{amsmath,epsfig,amssymb,algorithm,algpseudocode,amsthm,cite,url}

\usepackage[justification=centering]{caption}
\usepackage{subcaption}
\allowdisplaybreaks
\usepackage{csquotes}
\usepackage{verbatim}
\usepackage[english]{babel}
\usepackage{amsmath,amssymb}

\usepackage{verbatim}

\newtheorem{theorem}{\bf Theorem}

\newtheorem{proposition}{\bf Proposition}

\newtheorem{remark}{Remark}

\usepackage{setspace}	
\begin{document}
%
%
\title{\setstretch{1.3}Unmanned Aerial Vehicle with Underlaid Device-to-Device Communications: Performance and Tradeoffs}
\author{\normalsize{Mohammad Mozaffari,~\IEEEmembership{\normalsize Student Member,~IEEE},
{Walid Saad,~\IEEEmembership{\normalsize Senior Member,~IEEE}},
{Mehdi Bennis,~\IEEEmembership{\normalsize Senior Member,~IEEE}},
 and {M\'{e}rouane Debbah,~\IEEEmembership{\normalsize Fellow,~IEEE}}}
\thanks{M.~Mozaffari and W.~Saad are with Wireless@VT, Department of ECE, Virgina Tech, Blacksburg, VA, USA. Emails: \{mmozaff,walids\}@vt.deu. M.~Bennis is with CWC - Centre for Wireless Communications, Oulu, Finland, Email: {bennis@ee.oulu.fi}. M.~Debbah is with Mathematical and Algorithmic Sciences Lab, Huawei France R \& D, Paris, France, Email:{merouane.debbah@huawei.com}.}
}
\maketitle\vspace{-0.8cm}
\begin{abstract}
In this paper, the deployment of an unmanned aerial vehicle (UAV) as a flying base station used to provide on the fly wireless communications to a given geographical area is analyzed. In particular, the co-existence between the UAV, that is transmitting data in the downlink, and an underlaid device-to-device (D2D) communication network is considered. For this model, a tractable analytical framework for the coverage and rate analysis is derived. Two scenarios are considered: a static UAV and a mobile UAV. In the first scenario, the average coverage probability and the system sum-rate for the users in the area are derived as a function of the UAV altitude and the number of D2D users. In the second scenario, using the disk covering problem, the minimum number of stop points that the UAV needs to visit in order to completely cover the area is computed. \textcolor{black}{Furthermore, considering multiple retransmissions for the UAV and D2D users, the overall outage probability of the D2D users is derived.} Simulation and analytical results show that, depending on the density of D2D users, optimal values for the UAV altitude exist for which the system sum-rate and the coverage probability are maximized. Moreover, our results also show that, by enabling the UAV to intelligently move over the target area, the total required transmit power of UAV while covering the entire area, is minimized. Finally,  in order to provide a full coverage for the area of interest, the tradeoff between the coverage and delay, in terms of the number of stop points, is discussed.
\end{abstract} \vspace{-0.4cm}

\section{Introduction}\vspace{-0.1cm}
The use of unmanned aerial vehicles (UAVs) as flying base stations that can boost the capacity and coverage of existing wireless networks has recently attracted significant attention \cite{Bucaille} and \cite{zhan}. One key feature of a UAV that can potentially lead to the coverage and rate enhancement is having line-of-sight (LoS) connections towards the users. Moreover, owing to their agility and mobility, UAVs can be quickly and efficiently deployed to support cellular networks and enhance their quality-of-service (QoS). 
On the one hand, UAV-based aerial base stations can be deployed to enhance the wireless capacity and coverage at temporary events or hotspots such as sport stadiums and outdoor events. On the other hand, they can be used in public safety scenarios to support disaster relief activities and to enable communications when conventional terrestrial networks are damaged \cite{Bucaille}. Another important application of UAVs is in the Internet of Things (IoT) in which the devices have small transmit power and may not be able to communicate over a long range. In this case, a UAV can provide a means to collect the IoT data from one device and transmit it to the intended receiver \cite{lien} and \cite{dhillon}. Last but not least, in regions or countries in which building a complete cellular infrastructure is very expensive, deploying UAVs is highly beneficial as it removes the need for towers and cables. In order to reap the benefits of UAV deployments for communication purposes, one must address a number of technical challenges that include performance analysis, channel modeling, optimal deployment, resource management, and energy efficiency, among others \cite{HouraniModeling, FengModelling, FengPath, Holis,HouraniOptimal,Mozaffari,kosmerl,Daniel,Rohde,Han,Jiang, Mozaffari2}.

The most significant existing body of work on UAV communications focuses on air-to-ground channel modeling \cite{HouraniModeling, FengModelling, FengPath, Holis}. For instance, in \cite{HouraniModeling} and \cite{FengModelling}, the probability of line of sight (LoS) for air-to-ground communication as a function of the elevation angle and average height of buildings in a dense urban area was derived. The air-to-ground path loss model has been further studied in \cite{FengPath} and \cite{Holis}. As discussed in \cite{Holis}, due to path loss and shadowing, the characteristics of the air-to-ground channel are shown to depend on the height of the aerial base stations. 

To address the UAV deployment challenge, the authors in \cite{HouraniOptimal} derived the optimal altitude enabling a single, static UAV to achieve a maximum coverage radius. However, in this work, the authors simply defined a deterministic coverage by comparing the path loss with a specified threshold and did not consider the coverage probability. The work in \cite{Mozaffari} extends the results of \cite{HouraniOptimal} to the case of two UAVs while considering interference between the UAVs. In \cite{kosmerl}, the authors studied the optimal placement of UAVs for public safety communications in order to enhance the coverage performance. However, the results presented in \cite{kosmerl} are based on simulations and there is no significant analytical analysis. Moreover, the use of UAVs for supplementing existing cellular infrastructure was discussed in \cite{Daniel} which provides a general view of practical considerations for integrating UAVs with cellular networks. The work in \cite{Rohde} considered the use of UAVs to compensate for the cell overload and outage in cellular networks. However, \cite{Daniel} does not provide any analysis on the coverage performance of UAVs and their optimal deployment methods. In \cite{Han}, the authors investigated how to optimally move UAVs for improving connectivity of ad hoc networks. However, \cite{Han} only focused on an ad-hoc network and assumed that the UAV have complete information about the location of nodes. In \cite{Jiang}, considering static ground users, the optimal trajectory and heading of UAVs equipped with multiple antennas for ground to air uplink scenario was derived. The work in \cite{Mozaffari2} proposed a power efficient deployment and cell association for multiple UAVs in downlink transmissions.   

For scenarios in which there is limited or no infrastructure support, beyond the use of UAVs, there has been considerable recent works that study the use of direct device-to-device (D2D) communications between wireless users over the licensed spectrum \cite{yaacoub}. Such D2D communications has been shown improve coverage and capacity of existing wireless networks, such as cellular systems. In particular, in hotspot areas or public safety scenarios, D2D will allow users to communicate directly with one another without significant infrastructure. D2D communications are typically deployed using underlaid transmission links which reuse existing licensed spectrum resources \cite{doppler}. Therefore, deploying a UAV over a spectrum band that must be shared with an underlaid D2D network will introduce important interference management challenges. In the literature, there are some studies on the coexistence of the underlaid D2D and cellular communications with a single base station \cite{lee}. Furthermore,  the authors in \cite{shalmashi}  exploited the interplay between the massive MIMO and underlaid D2D communications for a single cell case. The authors in \cite{lin} extended the previous work on the D2D/massive MIMO coexistence  to the multi-cell scenario. However, none of theses prior works studied the coexistence of UAVs and underlaid D2D communications. In particular, a comprehensive analytical analysis to evaluate this coexistence in terms of different performance metrics, such as coverage and rate, is lacking in the current state-of-the-art \cite{HouraniOptimal,Han,lee,shalmashi,lin}.

Compared to the previous studies on the coexistence of D2D and cellular networks such as \cite{shalmashi} and \cite{lin}, the presence of an aerial UAV base station along with D2D links introduces new challenges. First, the channel modeling between the UAV and ground users will no longer be a classical fading channel, instead, it will be based on probabilistic LoS and NLoS links  \cite{HouraniModeling, FengModelling}, while the channel between a base station and the users will still follow a Rayleigh fading model. Second, unlike conventional, fixed base stations, the height of a UAVs is adjustable and this impacts the channel characteristics and the coverage performance. Third, the potential mobility of a UAV introduces new dimensions to the problem and the impact of such mobility on D2D and network performance must be analyzed. The prior studies on UAVs such as \cite{HouraniOptimal,HouraniModeling, FengModelling, FengPath, Holis,Mozaffari,kosmerl,Daniel,Rohde,Han} have not addressed the third challenge. More specifically, the interplay between UAVs and D2D communications and the existing challenges and tradeoffs have not been investigated in these literature. To our best knowledge, this paper will  provide \emph{the first comprehensive fundamental analysis} on the performance of UAV communication in the presence of underlaid D2D links.       
 
\textcolor{black}{The main contribution of this paper is to analyze the coverage and rate performance of UAV-based wireless communication in the presence of underlaid D2D communication links. In particular, we consider a network in which a single UAV must provide downlink transmission support to a number of users within a given area. In this area, a subset of the devices is also engaged in D2D transmissions that  operate in an underlay fashion over the UAV's transmission. We consider two types of users, namely downlink users (DUs) which receive data from the UAV, and D2D users which communicate directly with one another. Here, the UAV must communicate with the DUs while taking into account the potential interference stemming from the underlaid D2D transmissions. For this network, we analyze two key cases: static UAV and mobile UAV. Using tools from stochastic geometry, for both scenarios, we derive the average downlink coverage probabilities for DUs and D2D users and we analyze the impact of the UAV altitude and density of the D2D users on the overall performance. For the static case, we find the optimal values for the UAV altitude which leads to a maximum coverage probability for DUs. In addition, considering both DUs and D2D users, an optimal altitude which maximizes the system sum-rate is computed. Our results demonstrate that the optimal UAV altitude decreases as the density of D2D users increases. The results show that a maximum system sum-rate can be achieved if the UAV altitude is appropriately adjusted based on the D2D users' density. Furthermore, for a given UAV altitude, we show that  an optimal value for the number of D2D users that maximizes the system sum-rate exists.} 

\textcolor{black}{For the mobile UAV case, we assume that the UAV can travel over the area while stopping at some given locations in order to serve the downlink users. Using the disk covering problem, we find a minimum number of stop points that the UAV needs to to completely cover the area. This can be interpreted as the fastest way to cover the whole area with a minimum required transmit power.  In addition, we analyze the tradeoff between the number of stop points, which is considered as delay here, and the coverage probability for the downlink users.  Moreover,  considering retransmissions at different time instances, we derive the overall outage probability for the D2D communications. We show that, in order to enhance the coverage for DUs, the UAV should stop in more locations over the target area which can, in turn, lead an increased delay for DUs and higher outage probability for D2D users. For example, our results show that for a given density of D2D users, to increase the DU coverage probability from 0.4 to 0.7, the number of stop points should be increased from 5 to 23. Furthermore, the number of stop points is shown to significantly depend on the number of D2D users. For instance, if the average number of D2D users in the area increases from $50$ to $100$, in order to maintain the DUs' coverage requirement, the number of stop points should be increased from 20 to 55.}
 
The rest of this paper is organized as follows. Section II presents the system model and describes the air-to-ground channel model. In Section III, coverage probabilities for DUs and D2D users are provided for a single static UAV. Section IV presents the performance evaluation for one mobile UAV which is used to provide full coverage for the target area. Section V presents the simulation results while Section VI draws some conclusions.

\section{System Model}\label{sec:sysmodel}
\textcolor{black}{Consider a circular area with a radius ${R_c}$ in which a number of wireless users are deployed. In this area, a UAV (at low altitude platform) is deployed to act as a flying base station and serve a subset of those users.} In this network, the users are divided into two groups: downlink users located uniformly in the cell with density $\lambda_{du}$ (number of users per $\rm{m}^{2})$, \textcolor{black}{ and D2D users whose distribution follows homogeneous Poisson point processes (PPP)}  ${\Phi _{\rm{B}}}$  \cite{haenggi} with density of  ${\lambda _{d}}$ (number of D2D pairs per $\rm{m}^{2})$. Note that, the average number of users in a given area is equal to the density of the users multiplied by the size of the area. Here, we focus on the downlink scenario for the UAV and we assume that the D2D users communicate in an underlay fashion. Furthermore, we assume that a D2D receiver connects to its corresponding D2D transmitter pair located at a fixed distance away from it in an isotropic direction \cite{lee}. Therefore, the received signals at the D2D receiver include the desired signal from the D2D transmitter pair and interference from the UAV and other D2D transmitters. A downlink user, on the other hand, receives the desired signal from the UAV but it also experiences interference from all the D2D transmitters.  \textcolor{black}{It should be noted that, in our model, the UAV provides service for downlink users (DUs) located inside a given, finite area with radius $R_c$. Nonetheless, we assume that the D2D users are spatiality distributed according to a PPP over an infinite area. In other words, each user receives interference from an infinite number of D2D transmitters. This is a typical assumption in PPP analysis which ensures that, the average amount of received interference from D2D transmitters does not depend on the location of the users \cite{shalmashi, martin}, and \cite{baccelli}.} 

\begin{figure}[t!]
  \begin{center}
    \includegraphics[width=0.7\textwidth]{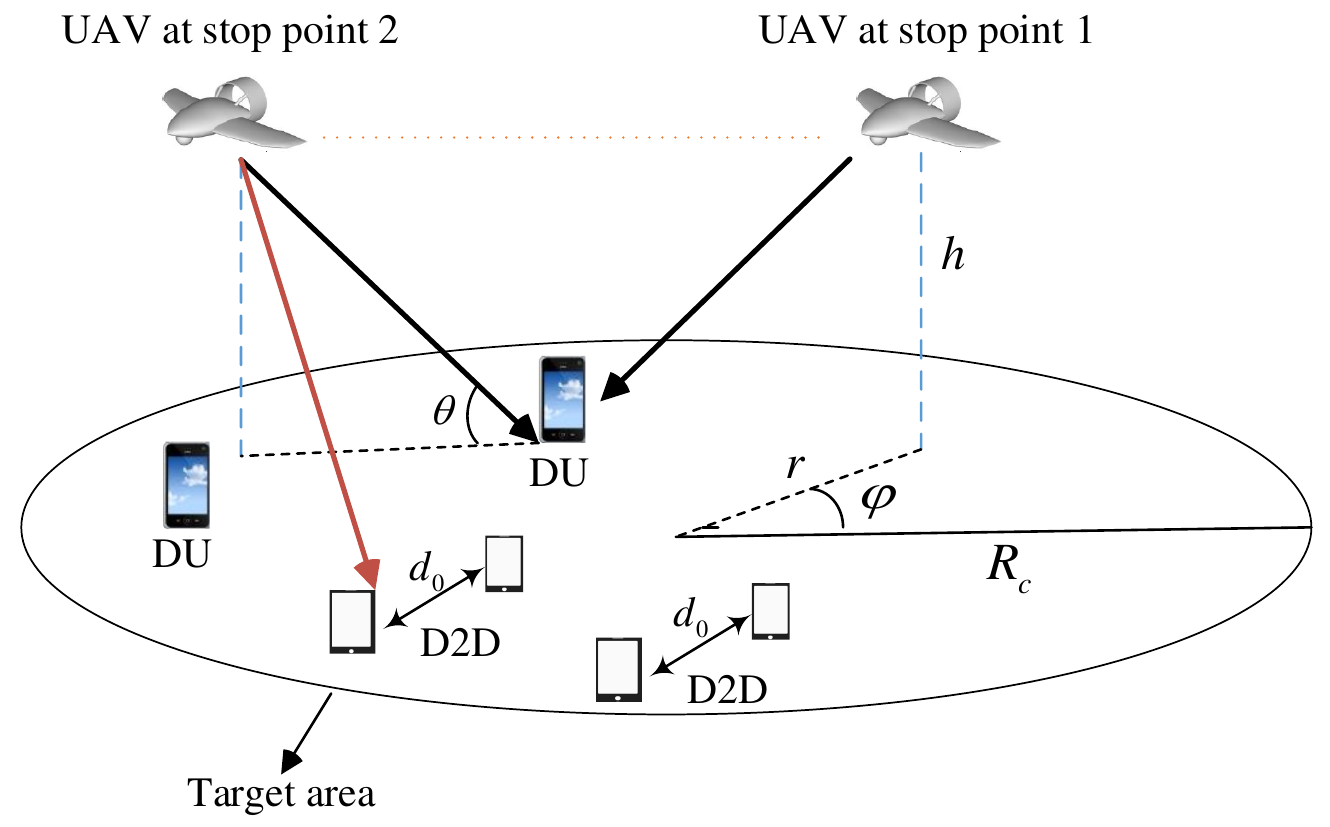}
    \caption{\label{fig: Fig1} Network model including a UAV, downlink users and D2D.}
  \end{center}
\end{figure}

The signal to interference plus noise ratio (SINR) expression for a D2D receiver is:
\begin{equation}
{\gamma _d} = \frac{{{P_{r,d}}}}{{I_d^c + {I_u} + N}},
 \end{equation}
 where ${P_{r,d}}$ is the received signal power from the D2D transmitter,  $I_d^c$ is the total interference from other D2D users, ${I_u}$  is the interference from the UAV, and $N$ is the noise power. Moreover, we have:
 \begin{gather}
{P_{r,d}} = {P_d}{d_{0}^{ - {\alpha _d}}} g_{0},\\
I_d^c = \sum\limits_{i \ne 0} {{P_d}{d_{i}}^{ - {\alpha _d}}{g_{i}}}, \\
{I_d} = \sum\limits_i {{P_d}{d_{i}}^{ - {\alpha _d}}{g_{i}}}, 
 \end{gather}
where the index $i=0$ is used for the selected D2D transmitter/receiver pair, $g_{0}$ and $g_{i}$ are, respectively, the channel gains between a D2D receiver and its corresponding D2D transmitter, and the ${i^{th}}$ interfering D2D transmitters. For the D2D transmission, we assume a Rayleigh fading channel model \cite{lee,lin} and \cite{Afshang}. ${{P_d}}$ is the D2D transmit power which is assumed to be fixed and equal for all the users, ${{d_{i}}}$ is the distance between a D2D receiver and the ${i^{th}}$ D2D transmitter, $d_{0}$ is the fixed distance between the D2D receiver and transmitter of the selected D2D pair, and  ${\alpha _d}$  is the path loss exponent between D2D users. Note that the received signal powers as well as the noise power are normalized by a path loss coefficient. 
 
The SINR expression for a DU user that connects to the UAV is:

\begin{equation}
{\gamma _u} = \frac{{{P_{r,u}}}}{{{I_d} + N}},
 \end{equation}
where ${P_{r,u}}$ is the received signal power from  the  UAV and ${I_d}$ is the total interference power from D2D transmitters. 
\textcolor{black}{
Finally, the SINR-based coverage probability for the downlink users and the D2D users is given by:
\begin{align}
{P_{{\mathop{\rm cov}} ,{{du}}}}(\beta ) &={\mathds {P}}\left[ {\gamma _u} \ge \beta \right],\\
{P_{{\mathop{\rm cov}} ,{{d}}}}(\beta ) &={\mathds {P}}\left[ {\gamma _d} \ge \beta \right],
\end{align}
where $\gamma _u$ and $\gamma _d$ are, respectively, the SINR values at the location of the downlink users and the D2D users, and $\beta$ is the SINR threshold.
}
 
\subsection{Air-to-ground channel model}
\textcolor{black}{
As discussed in \cite{HouraniModeling} and \cite{HouraniOptimal}, the ground receiver receives three groups of signals including LoS, strong reflected non-line-of-sight (NLoS) signals, and multiple reflected components which cause multipath
fading. These groups can be considered separately with different probabilities of occurrence as shown in \cite{Holis} and \cite{HouraniModeling}. Typically, it is assumed that the received signal is categorized in only one of those groups \cite{HouraniOptimal}. Each group has a specific probability of occurrence which is a function of environment, density and height of buildings, and elevation angle. Note that the probability of having the multipath fading is significantly lower than the LoS and NLoS groups \cite{HouraniOptimal}. Therefore, the impact of small scale fading can be neglected in this case \cite{HouraniModeling}.} One common approach to modeling air-to-ground propagation channel is to consider LoS and NLoS components along with their occurrence probabilities separately as shown in \cite{HouraniModeling} and \cite{Holis}. Note that for NLoS connections due to the shadowing effect and the reflection of signals from obstacles, path loss is higher than in LoS. Hence, in addition to the free space propagation loss, different excessive path loss values are assigned to LoS and NLoS links.
Depending on the LoS or NLoS connection between the user and UAV, the received signal power at the user location is given by \cite{HouraniOptimal}: \vspace{-0.3cm}

\begin{equation}
{P_{r,u}} = \left\{ \begin{array}{l}
{P_u}{\left| {{X_{u}}} \right|^{ - {\alpha _u}}}{\rm{{\rm \hspace*{1.7cm}{LoS\hspace*{.2cm} connection}}}},\\
\eta {P_u}{\left| {{X_{u}}} \right|^{ - {\alpha _u}}}{\rm \hspace*{1.3cm} {   NLoS \hspace*{.2cm}connection}},
\end{array} \right.
\end{equation}
where  ${P_u}$  is the UAV transmit power, $\left| {{X_{u}}} \right|$  is the distance between a generic user and the UAV, ${\alpha _u}$  is the path loss exponent over the user-UAV link, and  $\eta $ is an additional attenuation factor due to the NLoS connection. Here, the probability of LoS connection depends on the environment, density and height of buildings, the location of the user and the UAV, and the elevation angle between the user and the UAV. The LoS probability can be expressed as follows  \cite{HouraniOptimal}:
\begin{equation}
{P_{{\rm{LoS}}}} = \frac{1}{{1 + C\exp ( - B\left[ {\theta  - C} \right])}},
\end{equation}
where $C$  and $B$  are constant values which depend on the environment (rural, urban, dense urban, or others) and $\theta$  is the elevation angle.  Clearly, ${\theta} = \frac{{180}}{\pi } \times {\sin ^{ - 1}}\left( {{\textstyle{{{h}} \over {\left| {{X_{u}}} \right|}}}} \right)$, $\left| {{X_u}} \right| = \sqrt {{h^2} + {r^2}}$ and also, probability of NLoS is ${{P}_{{\text{NLoS}}}} = 1 - {{P}_{{\text{LoS}}}}$.\\
As observed from (9), the LoS probability increases as the elevation angle between the user and UAV increases. 

\textcolor{black}{Given this model, we will consider two scenarios: \emph {a static UAV} and \emph {a mobile UAV}. For each scenario, we will derive the coverage probabilities and average rate for DUs and D2D users. Once those metrics are derived, considering the D2D users density, we obtain optimal values for the UAV altitude that maximize the coverage probability and average rate.}
 
\section{Network with a Static UAV}\label{sec:da}\vspace{-0.1cm}
\textcolor{black}{	
In this section, we evaluate the coverage performance of the scenario in which one UAV located at the altitude of $h$ in the center of the area to serve the downlink users in the presence of underlaid D2D communications. It can be shown that, for a uniform distribution of users over the given area, palcing the UAV in the center of the area can maximize the coverage probability of the downlink users.\vspace{-0.2cm}
\subsection{Coverage probability for D2D users}
Consider a D2D receiver located at $(r,\varphi )$, where $r$ and $\varphi$ are the radius and angle in a polar coordinate system assuming that the UAV is located at the center of the area of interest. The distance between the D2D transmitter and its corresponding  receiver is fixed and it is denoted by $d_0$. 
 In this case, for underlaid D2D communication, the coverage probability for the D2D users  can be derived as follows:}
\textcolor{black}{	
\begin{theorem}\label{T1}
The coverage probability for a D2D receiver, at the location $(r,\phi)$, connecting to its D2D transmitter located at a distance $d_0$ away from it, is given by: \\
\begin{align} \label{eq:Pd}
{P_{{\mathop{\rm cov}} ,d}}(r,\varphi ,\beta ) &= \exp \left( {\frac{{ - 2{\pi ^2}{\lambda _d}{\beta ^{2/{\alpha _d}}}{d_0^2}}}{{{\alpha _d}\sin (2\pi /{\alpha _d})}} - \frac{{\beta {d_0^{{\alpha _d}}}N}}{{{P_d}}}} \right)\nonumber\\
& \times \left( {{P_{{\rm{LoS}}}}\exp \left( {\frac{{ - \beta {d_0^{{\alpha _d}}}{P_u}{|X_u|^{ - {\alpha _u}}}}}{{{P_d}}}} \right) + {P_{{\rm{NLoS}}}}\exp \left( {\frac{{ - \beta {d_0^{{\alpha _d}}}\eta {P_u}{|X_u|^{ - {\alpha _u}}}}}{{{P_d}}}} \right)} \right),
\end{align}
where $\left| {{X_u}} \right| = \sqrt {{h^2} + {r^2}}$.\vspace{-0.3cm}
\end{theorem} \vspace{0.3cm}
}\vspace{-0.8cm}
\textcolor{black}{
\begin{proof}
See Appendix A.
\end{proof}
}
\textcolor{black}{
From this theorem, we can make several key observations. First, considering the fact that the UAV creates interference on the D2D users, increasing the UAV altitude to increase its distance from the D2D users  does not necessarily reduces the interference on the D2D users. As will be shown later by numerical simulations, by increasing the UAV altitude the D2D coverage probability decreases first, and then increases. This is due to the fact that, considering (9) and (10), although increasing the UAV altitude increases the path loss term, it also leads to a LoS probability. In general, the D2D users prefer to have the NLoS view towards the UAV and have a maximum distance from it, however, these two objectives conflicts with each other. Second, increasing the D2D transmit power ($P_d$), always enhances the D2D coverage probability, even in an interference limited scenario where noise is ignored. Typically, in the interference limited scenarios, increasing the transmit power of the D2D users does not improve the coverage performance due to the increased interference from other D2D transmitters. According to Theorem 1, although in the interference limited scenario ($N=0$) the first multiplying term in (10) is independent of $P_d$ due to the interference from D2D transmitters, the second term is an increasing function of $P_d$.  Finally, the D2D coverage probability in (10) decreases when the UAV transmit power increases. To cope with this situation, the D2D users can increase their transmit power or reduce the fixed distance parameter ($d_0$). In addition, decreasing the D2D user density improves the coverage probability due to decreasing the interference.}
\textcolor{black}{Note that the result presented in Theorem 1 corresponds to the coverage probability for a D2D user located at $(r,\varphi )$. To compute the average coverage probability in the cell, we consider a uniform distribution of users} over the area with $f(r,\varphi ) = \frac{r}{{\pi R_c^2}},  \hspace*{.2cm}  {\rm{0}} \le r \le {R_c} \hspace*{.1cm},\hspace*{.1cm} {\rm{0}} \le \varphi  \le 2\pi$\footnote{Note that the number of users has a Poisson distribution but their location follows the uniform distribution over the area}, \textcolor{black}{where  $R_c$ is the radius of the desired circular area.} Then, we compute the average over the desired area. The average coverage probability for D2D users will be:  
\begin{align}
{{\bar P}_{{\mathop{\rm cov}} ,{d}}}(\beta )& = {{\mathds {E}}_{r,}}_\varphi \left[ {{P_{{\mathop{\rm cov}} ,d}}(r,\varphi ,\beta )} \right]\nonumber\\
& = \exp \left( {\frac{{ - 2{\pi ^2}{ \lambda_d} {\beta ^{2/{\alpha _d}}}{d_0^2}}}{{{\alpha _d}\sin (2\pi /{\alpha _d})}} - \frac{{\beta {d_0^{{\alpha _d}}}N}}{{{P_d}}}} \right)\int\limits_0^{{R_c}} {{{\mathds {E}}_{{I_u}}}\left[ {\exp (\frac{{ - \beta {d_0^{{\alpha _d}}}{I_u}}}{{{P_d}}})} \right]} f(r,\varphi )\text{d}r\text{d}\varphi \nonumber\\
& = \exp \left( {\frac{{ - 2{\pi ^2}{\lambda_d} {\beta ^{2/{\alpha _d}}}{d_0^2}}}{{{\alpha _d}\sin (2\pi /{\alpha _d})}} - \frac{{\beta {d_0^{{\alpha _d}}}N}}{{{P_d}}}} \right)\int\limits_0^{{R_c}} {{{\mathds {E}}_{{I_u}}}\left[ {\exp (\frac{{ - \beta {d_0^{{\alpha _d}}}{I_u}}}{{{P_d}}})} \right]} \frac{{2r}}{{{R_c^2}}}\text{d}r.
\end{align}

From (11), we can see that the average coverage probability for D2D users increases as the size of the area, $R_c$, increases. In fact, when the UAV serves a larger area, the average distance of D2D users from the UAV increases and on the average they receive lower interference from it. Next, we provide a special case for (11) in which the UAV has a very high altitude or very small transmit power.
\textcolor{black}{
\begin{remark}
For ${P_u}=0$ or $h\to\infty$, the average coverage probability for the D2D users is simplified to \cite{baccelli}:
\begin{equation}\label{cov}
{\bar P_{{\mathop{\rm cov}} ,{\rm{d}}}}(\beta ) = \exp \left( {\frac{{ - 2{\pi ^2}{ \lambda_d} {\beta ^{2/{\alpha _d}}}{d_0^2}}}{{{\alpha _d}\sin (2\pi /{\alpha _d})}} - \frac{{\beta {d_0^{{\alpha _d}}}N}}{{{P_d}}}} \right),
\end{equation}
Note that, (\ref{cov}) corresponds to the coverage probability in overlay D2D communication in which there is no interference between the UAV and the D2D transmitters. It should be noted that, this result is also related to the success probability in a bipolar ad hoc network \cite{baccelli}. 
\end{remark}
} \vspace{-0.3cm}

\subsection{Coverage Probability for Downlink Users } 
\textcolor{black}{
Here, we first derive the bound and lower bound for the downlink users' coverage probability.}\vspace{-0.7cm}
\textcolor{black}{
\begin{theorem}
The lower bound and upper bound of the average coverage probability for DUs in the area of interest is given by:
\begin{align}
{\bar P^L_{{\mathop{\rm cov}} ,\textnormal{du}}}(\beta,h ) & = \int\limits_0^{{R_c}} {{{{P}}_{{\rm{LoS}}}}(r,h)} {L_I}\left(\frac{{{P_u}{|X_u|}^{ - {\alpha _u}}}}{\beta } - N\right)\frac{{2r}}{{R_c^2}}{\rm{d}}r \nonumber\\ &+\int\limits_0^{{R_c}} {{{{P}}_{{\rm{NLoS}}}}(r,h)} {L_I}\left(\frac{{\eta {P_u}{|X_u|}^{ - {\alpha _u}}}}{\beta } - N\right)\frac{{2r}}{{R_c^2}}{\rm{d}}r,\\
{\bar P^U_{{\mathop{\rm cov}} ,\textnormal{du}}}(\beta,h ) &= \int\limits_0^{{R_c}} {{{{P}}_{{\rm{LoS}}}}(r,h)} {U_I}\left(\frac{{{P_u}{|X_u|}^{ - {\alpha _u}}}}{\beta } - N\right)\frac{{2r}}{{R_c^2}}{\rm{d}}r \nonumber\\ &+\int\limits_0^{{R_c}} {{{{P}}_{{\rm{NLoS}}}}(r,h)} {U_I}\left(\frac{{\eta {P_u}{|X_u|}^{ - {\alpha _u}}}}{\beta } - N\right)\frac{{2r}}{{R_c^2}}{\rm{d}}r,
\end{align}
\end{theorem}
where $\beta N < {P_u}||{X_u}|{|^{ - {\alpha _u}}}$, and for any $T>0$, \begin{equation}
\begin{small}
{L_I}(T) = \left( {1 - \frac{{2\pi { \lambda_d} {\rm{ }}\Gamma (1 + 2/{\alpha _d})}}{{{\alpha _d} - 2}}{{\left( {\frac{T}{{{P_d}}}} \right)}^{ - 2/{\alpha _d}}}} \right)\exp\left( { - \pi{ \lambda_d} {{\left( {\frac{T}{{{P_d}}}} \right)}^{ - 2/{\alpha _d}}}\Gamma (1 + 2/{\alpha _d})} \right),
\end{small}
\end{equation}
\begin{equation}
{U_I}(T) =  \exp\left( { - \pi { \lambda_d} {{\left( {\frac{T}{{{P_d}}}} \right)}^{ - 2/{\alpha _d}}}\Gamma (1 + 2/{\alpha _d})} \right).
\end{equation}
Also, $\Gamma (t) = \int\limits_0^\infty  {{x^{t - 1}}{e^{ - x}}} \textnormal{d}x$ is the gamma function \cite{artin}.
\begin {proof}
See Appendix B.
\end{proof}
}
From Theorem 2, we can first see that, for $T >> {P_d}$, given that ${e^{ - x}} \approx 1 - x$ when $x \to 0$, we have ${U_I}(T) = {L_I}(T) \approx 1 - \pi {\lambda _d}{\left( {\frac{T}{{{P_d}}}} \right)^{ - 2/{\alpha _d}}}\Gamma (1 + 2/{\alpha _d})$. This means that the lower bound and upper bound become tighter for lower transmit power of D2D users. Moreover, from (15) and (16), when ${\lambda _d} \to\infty$, the number of D2D users tends to infinity and ${U_I}={L_I}=0$. Consequently, the downlink users experience an infinite interference from the D2D users which results in ${{\bar P}_{{\mathop{\rm {cov}}} ,\rm{du}}}=0$.

\textcolor{black}{
Furthermore, considering (9), (13), and (14), we can see that increasing the UAV altitude ($h$), can enhance the LoS probability and the coverage probability. On the other hand, due to increasing $|X_u|$, $L_I$ and $U_I$ decrease, and hence the coverage probability for downlink users decreases. Therefore, in order to achieve the maximum coverage, the altitude of the UAV should be carefully adjusted. 
}

As per Theorem 2, increasing ${R_c}$ decreases the average coverage probability for the downlink users. However, higher ${R_c}$ results in a higher D2D average coverage probability. Moreover, the average coverage probability for downlink users decreases as the density of the D2D users increases. In this case, to improve the  DUs coverage performance, one must increase ${P_u}$ or reduce ${R_c}$. Next, we derive the DU coverage probability in the absence of the D2D users.\vspace{-0.03cm}    

\begin{proposition}
For low density and transmit power of D2D users, the interference from D2D users is negligible compared to the UAV, then, the exact average coverage probability for the downlink users can be expressed by: 
\begin{equation}
{{\bar P}_{{\mathop{\rm cov}} ,\rm{du}}}(\beta ) =\int_0^{{\textnormal{min}}[{{(\frac{{{P_u}}}{{\beta N}})}^{{1/{\alpha _u}}}},{R_c}]} {{{{P}}_{{\rm{LoS}}}}(r)\frac{{2r}}{{R_c^2}}{\rm{ }}} \text{d}r + \int_0^{{\textnormal{min}}[{{(\frac{{{\eta P_u}}}{{\beta N}})}^{{1/{\alpha _u}}}},{R_c}]} {{{{P}}_{{\rm{NLoS}}}}(r)\frac{{2r}}{{R_c^2}}{\rm{ }}}\textnormal{d}r.
\end{equation}
\end{proposition}
\begin{proof}
For a DU located at $(r,\varphi)$, the coverage probability in absence of D2D users becomes
\begin{align}
{P_{{\mathop{\rm cov}} ,\rm{du}}}(r,\varphi ,\beta ) &= \mathds {P}\left[ {\gamma _u} \ge \beta \right]  = {{{P}}_{{\rm{LoS}}}}(r)\mathds {P}\left[ {\gamma _u} \ge \beta |{\rm{LoS}}\right]  + {{{P}}_{{\rm{NLoS}}}}(r)\mathds {P}\left[ {\gamma _u} \ge \beta |{\rm{NLoS}}\right] \nonumber\\
 &= {{{P}}_{{\rm{LoS}}}}(r)\mathds{1}\left[ r \le {\left( {\frac{{{P_u}}}{{\beta N}}} \right)^{1/{\alpha _u}}}\right]  + {{{P}}_{{\rm{NLoS}}}}(r)\mathds{1}\left[ r \le {\left( {\frac{{\eta {P_u}}}{{\beta N}}} \right)^{1/{\alpha _u}}}\right],
 \end{align}
 The average coverage probability is computed by taking the average of ${P_{{\mathop{\rm cov}} ,\rm{du}}}(r,\varphi ,\beta )$ over the cell with the radius  ${R_c}$.
\begin{align}
{P_{{\mathop{\rm cov}} ,{du}}}(r,\varphi ,\beta ) &= {{\mathds {E}}_{r,}}_\varphi \left[ {{P_{{\mathop{\rm cov}} ,{du}}}(r,\varphi ,\beta )} \right]\nonumber\\
&=\int_0^{{\textnormal{min}}[{{(\frac{{{P_u}}}{{\beta N}})}^{{1/{\alpha _u}}}},{R_c}]} {{{{P}}_{{\rm{LoS}}}}(r)\frac{{2r}}{{R_c^2}}{\rm{ }}} \text{d}r + \int_0^{{\textnormal{min}}[{{(\frac{{{\eta P_u}}}{{\beta N}})}^{{1/{\alpha _u}}}},{R_c}]} {{{{P}}_{{\rm{NLoS}}}}(r)\frac{{2r}}{{R_c^2}}{\rm{ }}}\textnormal{d}r.
\end{align}

%
\end{proof}

Proposition 1 gives the exact expression for the downlink users' coverage probability when the interference from D2D users, due to their low density and low transmit power, is negligible compared to the UAV. Therefore, the result in Proposition 1 shows the maximum achievable coverage performance for downlink users when the received signal from the UAV is dominant compared to the interference from the D2D transmitters.
\subsection{System sum-rate}
\textcolor{black}{
Now, we investigate the average achievable rates for the DUs and D2D users which can be expressed as in \cite{shalmashi}: 
 \begin{gather}
{{\bar C}_{du}} = W{\log _2}(1 + \beta ){{\bar P}_{{\mathop{\rm cov}} ,du}}(\beta ),\\
{{\bar C}_d} = W{\log _2}(1 + \beta ){{\bar P}_{{\mathop{\rm cov}} ,d}}(\beta ),
\end{gather}
where  $W$  is the transmission bandwidth.}
\textcolor{black}{
Considering the whole DUs and D2D users in the cell, the system sum-rate, ${\bar C_{{\rm{sum}}}}$, can be derived as a function of the coverage probabilities and the number of users as follows:\vspace{-0.5cm}}

\begin{equation}
{\bar C_{{\rm{sum}}}} = {R_c}^2\pi {\lambda _{du}}{\bar C_{du}} + {R_c}^2\pi {\lambda _d}{\bar C_d}. 
\end{equation}
Assuming  $\mu  = \frac{{{\lambda _{du}}}}{{{\lambda _d}}}$, we have
\begin{equation}
{\bar C_{{\rm{sum}}}} = {\lambda _d}{R_c}^2\pi \left[ {\mu {{\bar P}_{{\mathop{\rm cov}} ,du}}(\beta ) + {{\bar P}_{{\mathop{\rm cov}} ,d}}(\beta )} \right]W{\log _2}(1 + \beta ),
\end{equation}
where ${R_c}^2\pi {\lambda _d}$ and ${R_c}^2\pi {\lambda _{du}}$ are the number of DUs and D2D users in the target area respectively.

From (30), observe that, on the one hand,  ${\bar C_{{\rm{sum}}}}$ is directly proportional to  ${\lambda _d}$, but on the other hand, it depends on the coverage probabilities of DUs and D2D users which both are decreasing  functions of D2D user density. Therefore, in general increasing ${\lambda _d}$ does not necessarily enhance the rate. Note that, considering (11), (13), (14), and (23), for both ${\lambda _d} \to 0$ and ${\lambda _d} \to \infty$ cases the system sum-rate tend to zero. Hence, there is an optimum value for ${\lambda _d}$ that maximizes ${\bar C_{{\rm{sum}}}}$.

According to (23), ${\bar C_{{\rm{sum}}}}$ is a function of the coverage probability and a logarithmic function of the threshold ($\beta$). The former is a decreasing function of $\beta$ whereas the latter is  an increasing function of $\beta$. In other words, although increasing the threshold is desirable for the rate due to increasing the logarithmic function, it also reduces the coverage probability. Therefore, in order to achieve a maximum rate, a proper value for the threshold can be adopted. It should be noted that, the SINR threshold, $\beta$, is typically fixed and cannot be set lower than the receiver sensitivity. However, the analysis of different values of $\beta$ brings value in order to understand how one could change the SINR threshold value (in the future) through proper resource allocation or just system design (change in the number of users, etc).

\section{Network with a Mobile UAV}

Now, we assume that the UAV can move around the area of radius $R_c$ in order to provide coverage for all the downlink users in the target area. In particular, we consider a UAV that moves over the target area and only transmits at a given geographical location (area) which we hereinafter refer to as \enquote{stop points}. Each stop point represents a location over which the UAV stops and serves the present downlink users. 
Here, our first goal is to minimize the number of stop points (denoted by $M$) and determine their optimal location. Note that, as the UAV moves, it can have a different channel to a user at different time instances. The objective of the UAV  is to cover the entire area and ensure that the coverage requirements for all DUs are satisfied with a minimum UAV transmit power and minimum number of stop points. In other words, we find the minimum number and location stop points for the UAV to completely cover the area. We model this problem by exploiting the so-called \emph{disk covering problem} \cite{kershner}. In the disk covering problem, given a unit disk, the objective is to find the smallest radius required for $M$ equal smaller disks to completely cover the unit disk. In the dual form of the problem, for a given radius of small disks, the minimum number of disks required to cover the unit disk is found.

In Figure 2, we provide an illustrative example to show the mapping between the mobile UAV communication problem and the disk covering problem. In this figure, the center of small disks can be considered as the location of stop points and the radius of the disk is the coverage radius of the UAV. Using the disk covering problem analysis, in Table I, we present, for different number of stop points, the minimum required coverage radius of a UAV for completely covering the target area  \cite{kershner,toth}. Thereby, using the dual disk covering problem, for a given maximum coverage radius of a UAV, we can find the minimum number of stop points for covering the entire area. The detailed steps for finding the minimum number of stop points are provided next. 

First, we compute the coverage radius of the UAV based on the minimum requirement for the DU coverage probability. The coverage radius is defined as the maximum radius within which the coverage probability for all DUs (located inside the coverage range) is greater than a specified threshold, $\epsilon$. In this case, the UAV satisfies the coverage requirement of each DU which is inside its coverage range. The maximum coverage radius for the UAV at an altitude $h$ transmitting with a power ${P_u}$  will be given by:
 \begin{equation}
 {R_m} = \max \{ R|{P_{{\mathop{\rm cov}},du}}(\beta ,R) \ge \varepsilon ,{P_u},h\}  = P_{{\mathop{\rm cov}},du}^{ - 1}(\beta ,\varepsilon ),
 \end{equation}
where $\varepsilon$ is the threshold for the average coverage probability in the cell (area covered by the UAV).
Note that, a user is considered to be in coverage if it is in the coverage range of the UAV.
The minimum required number of stop points for the full coverage is:\vspace{-0.2cm}

 \begin{equation}
\left\{ \begin{array}{l}
L=\min \{ M\},\vspace{-0.1cm}\\
{P_{{\mathop{\rm cov}} ,du}}(r,\varphi ,\beta ) \ge \varepsilon,
\end{array} \right.
\end{equation}
where $M$ represents the number of stop points, the second condition guarantees that the area is completely covered by the UAV, and $L$ is the minimum value for the number of stop points if the following condition holds:
\begin{equation}
{R_{\min ,L}} \le {R_m} \le {R_{\min ,L-1}} \to \min \{ M\}  = L.
\end{equation}
By using Table I, we see that, ${R_{\min ,L - 1}}$ and ${R_{\min ,L }}$ are, respectively, the minimum radius required to cover the entire target area with $L-1$ and $L$ disks.   
\textcolor{black}{
After finding the minimum $M$, we can reduce the UAV transmission power such that the coverage radius decreases to the minimum required radius (${R_{\min, L }}$). In this way, the UAV transmit power is minimized. Thus we have
 \begin{equation} 
 {P_{u,\min }} = \mathop {{\mathop{\rm argmin}\nolimits} }\limits_{{P_u}} \{ P_{{\mathop{\rm cov}},\rm{du}}^{ - 1}(\beta ,\varepsilon ) = {R_{\min, L }}|h\},
 \end{equation}
where ${P_{u,\min }}$ is the minimum UAV transmit power.
Thereby, the minimum number of stop points leads to a full coverage at a minimum time with a  minimum  required transmit power.}

In summary, the proposed UAV deployment method that leads to the complete coverage with a minimum time and transmission power proceeds as follows. First, depending on the parameters of the problem such as density of users and threshold, we compute the maximum coverage radius of a UAV at the optimal altitude that can serve the DUs. Second, considering the size of target area, using the disk covering problem, we find the minimum required number of transmission points along with the coverage radius at each point. Third, we reduce the transmission power of UAV such that its maximum coverage radius becomes equal to the required coverage radius found in the previous step. Using the proposed method, the target area can be completely covered by the UAV with a minimum required transmit power and minimum number of stop points. \textcolor{black}{Next, we investigate the impact of the number of stop points on the full coverage time of the downlink users, and the overall outage probability of the D2D users.}


\begin{figure}[!t]
  \begin{center}
   \vspace{-0.2cm}
    \includegraphics[width=7cm]{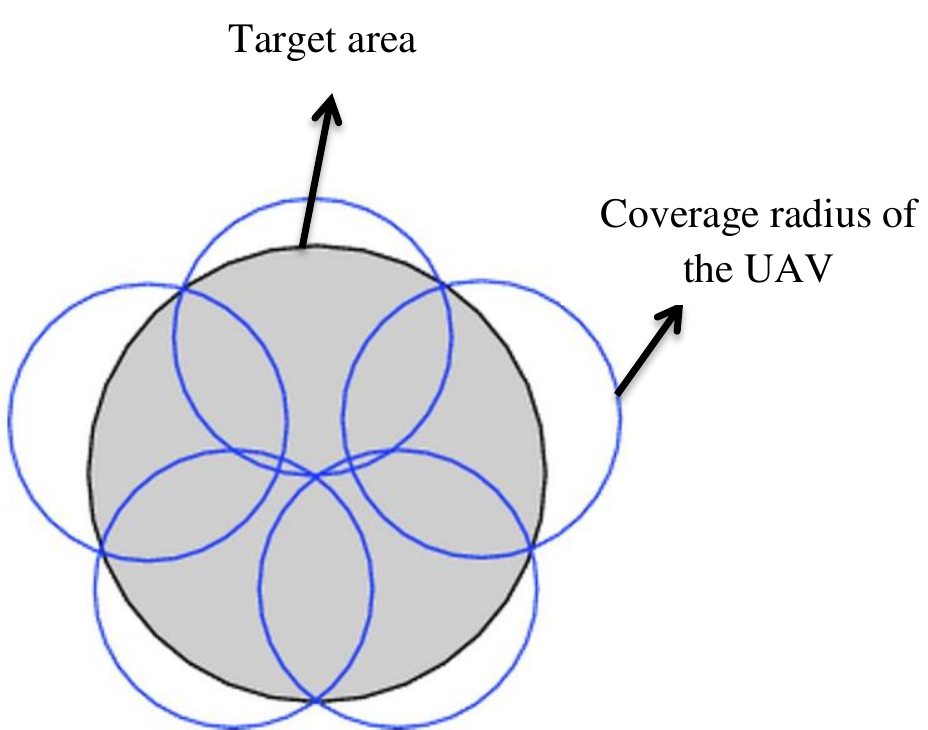}
    \vspace{-0.1cm}
    \caption{\label{fig: Fig1} Five disks covering problem.}
  \end{center}\vspace{-0.2cm}
\end{figure}

\begin{table}[!t]
\normalsize
\begin{center}
\captionof{table}{Number and radii of disks in the covering problem.}
\resizebox{11cm}{!}{
\begin{tabular}{|c|c|}

 \hline
 \textbf{Number of stop points} & \textbf { Minimum required coverage radius (${R_{\min }}$)}\\
\hline
$M=1,2$  & ${R_c}$   \\
\hline
$M=3$  & ${{\textstyle{{\sqrt 3 } \over 2}}R_c}$  \\
\hline
$M=4$  & ${\textstyle{{\sqrt 2 } \over 2}}{R_c}$  \\
\hline
$M=5$  & $0.61{\rm{ }}{R_c}$  \\
\hline
$M=6$  & $0.556{\rm{ }}{R_c}$ \\
\hline
$M=7$  & $0.5{\rm{ }}{R_c}$  \\
\hline
$M=8$  & $0.437{\rm{ }}{R_c}$ \\
\hline
$M=9$  & $0.422{\rm{ }}{R_c}$ \\
\hline
$M=10$  & $0.398{\rm{ }}{R_c}$ \\
\hline
$M=11$  & $0.38{\rm{ }}{R_c}$ \\
\hline
$M=12$  & $0.361{\rm{ }}{R_c}$ \\
\hline 
\end{tabular}\vspace{0.4cm} 
}
\end{center}
\end{table}
\textcolor{black}{
We consider the network during $M$ time instances in which the UAV and D2D users will execute $M$ retransmissions. Note that, our system model considers the downlink, therefore, the retransmissions are essentially from the UAV to the DUs, and from D2D transmitters to corresponding receivers. The moving UAV satisfies the coverage requirements of the downlink users in $M$ retransmissions from different locations. Clearly, as the number of stop points ($M$) increases, the time required for UAV to completely cover the desired area, increases. Here, the time that the UAV needs to provide the full coverage for the area by visiting all the stop points, is called delay. Hence, the delay depends on the travel time of the UAV between the stop points, and the time that UAV spends at each stop point for transmissions. Thus, the delay can be written as:\vspace{-0.3cm}
\begin{equation}
\tau  = {T_{tr}} + M{T_s}\vspace{-0.2cm}
\end{equation}
where $T_{tr}$ is the total UAV travel time, $M$ is the number of stop points, and $T_s$ is the time that the UAV stays at each stop point. Clearly, the travel time depends on the travel distance and location of the stop points, and the speed of the UAV. The total travel time will clearly increase as the number of stop points increases. However, in general, the exact relationship between $T_{tr}$ and $M$ strongly depends on the locations of the stop points which do not necessary follow a fixed path/distribution for different values of $M$. As an example, it can be shown that the exact travel time for $M=3$ and $M=4$ is $\frac{{\sqrt 3 {R_c}}}{v}$ and  $\frac{{3{R_c}}}{v}$ respectively, where $v$ is the speed of the UAV, and $R_c$ is the radius of the desired area.}
\textcolor{black}{
The residence time, $T_s$, depends on the multiple access method. If the UAV adopts a time division multiple access (TDMA) technique, the residence time will be a function of the number of stop points. In fact, a higher number of stop points corresponds to a smaller coverage region of the UAV. In this case, at each residence point, the UAV needs to provide service for a fewer number of users. Therefore, by increasing the number of stop points, the residence time can be decreased in the TDMA case. Considering a uniform distribution of the users, the residence time is approximately computed as:
\begin{equation}
{T_s} \approx {T_{s,1}}\frac{{R_{\min }^2(M)}}{{R_c^2}}U,
\end{equation}
where $T_{s,1}$ is the service time of UAV for each downlink user, $U$ is the number of downlink users, $R_{\min }$ is the coverage radius of the UAV which depends on $M$, the number of the stop points, and $R_c$ is the radius of the desired area.}
\textcolor{black}{However, if the UAV uses a frequency division multiple access (FDMA) technique, the users can be served simultaneously. In other words, the UAV does not need to use different time slots to serve the users. Therefore, if users are of homogeneous traffic type, the residence time of the UAV at each stop point does not depend on the number of the users, and hence it can be fixed. In this case, the residence time at each stop point will be constant and it does not depend on the coverage radius of the UAV and the number of stop points. As a result, ${T_s} = {T_{s,1}}$. In our model, we have considered FDMA for multiple access. Hence, the residence time is the same for all values of $M$.  In Figure \ref{delay}, we have shown the total delay versus the number of stop points for two values of residence time, and $v=10$\,m/s. As expected, the total delay increases as the number of stop points increases. Moreover, when the residence time of the UAV at each stop point increases, the additional delay due to a higher number of stop points increases. As we can see from Figure \ref{delay}, for $T_{s,1}=20$\,s, the delay increases from 230\,s to 480\,s if the number of stop points increase from 3 to 10. However, for $T_{s,1}=40$\,s the delay increases from 295\,s to 690\,s. Clearly, the delay and the number of stop points are directly related. It should be noted that, for our simulations, we consider the number of stop points as delay.}
\begin{figure}[!t]
  \begin{center}
   \vspace{-0.2cm}
   \includegraphics[width=9cm]{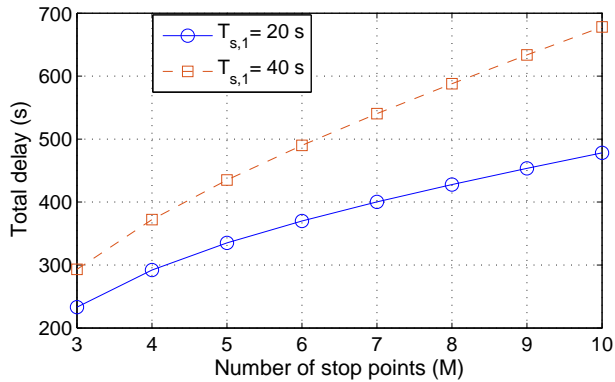}
    \vspace{-0.095cm}
    \caption{ Total delay increases as the number of stop points. \vspace{-.2cm}}
    \label{delay}
  \end{center}\vspace{-0.07cm}
\end{figure}

\textcolor{black}{Next, we derive the overall outage probability for a typical D2D user in the $M$ time instances for the mobile UAV case. The outage probability is the probability of having at least one failure during $M$ retransmissions. Assume that the relative location of the ${i^{th}}$ stop point with respect to the D2D user is $({r_i},{h_i})$  where ${r_i}$  is the distance between the projection of the UAV on the ground and D2D user and ${h_i}$  is the UAV altitude. Clearly, the distance between the user and UAV is $\left| {{X_{u,i}}} \right| = \sqrt {h_i^2 + {r_i}^2} $. For different time slots, the Rayleigh fading changes and can be considered independent \cite{martin}. However, since the location of the D2D users do not significantly change during the multiple time slots, the interference from the D2D users are correlated. Then, the overall outage probability for D2D users can be found in the next theorem. }
\textcolor{black}{
\begin {theorem}
The overall outage probability for D2D users in $M$ retransmissions considering moving UAV is given by:
\begin{small}
\begin{align}
{P_{out,d}}=1- \exp \left( { - {\lambda _d}\int_{{R^2}} {\left[ {1 - {{\left( {\frac{1}{{1 + \frac{{\beta |x{|^{ - {\alpha _d}}}}}{{d_0^{ - {\alpha _d}}}}}}} \right)}^M}} \right]{\rm{d}}x} } \right) \prod\limits_{i = 1}^M {{\mathds {E}_{{I_{u,i}}}}\left[ {\exp \left( {\frac{{ - d_0^{{\alpha _d}}\beta {I_{u,i}}}}{{{P_d}}}} \right)} \right]}\small {\exp \left( {\frac{{ - d_0^{{\alpha _d}}\beta MN}}{{{P_d}}}} \right)},
\end{align}
\end{small}where $M$ is the number of retransmissions, $I_{u,i}$ is the interference from the UAV at $i^{th}$ retransmission, and $E_{I_{u,i}}$(.) is: \\
${{\mathds {E}_{{I_{u,i}}}}\left[ {\exp \left( {\frac{{ - d_0^{{\alpha _d}}\beta {I_{u,i}}}}{{{P_d}}}} \right)} \right]}={{P_{{\rm{LoS,}}i}}\exp \left( {\frac{{ - \beta d_0^{{\alpha _d}}{P_u}|{X_{u,i}}{|^{ - {\alpha _d}}}}}{{{P_d}}}} \right) + {P_{{\rm{NLoS,}}i}}\exp \left( {\frac{{ - \beta d_0^{{\alpha _d}}\eta {P_u}|{X_{u,i}}{|^{ - {\alpha _d}}}}}{{{P_d}}}} \right)}$.\vspace{-0.2cm}
\end{theorem}
}\vspace{-0.5cm}
\textcolor{black}{
\begin{proof}
See Appendix C.
\end{proof}
From Theorem 3, we can observe that, increasing $M$ leads to a higher outage probability. In fact, as the number of stop points increases, the UAV creates a stronger interference on the D2D users. Consequently, $P_{out,d}$ tends to 1 for $M \to \infty$. However, the higher number of stop points for UAV enhances the coverage performance of the downlink users. Hence, a tradeoff between coverage performance of downlink users and the outage of D2D communications should be taking into account. Moreover, Theorem 3 shows that, in order to guarantee that the outage probability does not exceed a specified threshold for different values of $M$, we should adaptively reduce the distance between the D2D transmitter and receiver ($d_0$), or have orthogonal spectrum. 
}\vspace{-0.7cm}

\section{Simulation Results and Analysis}

\subsection{The static UAV scenario}
First, we compare our analytical results of the coverage probabilities using numerical simulations. Table II lists parameters used in the simulation and statistical analysis. These parameters are set based on typical values such as in \cite{HouraniOptimal} and \cite{shalmashi}. Here, we will analyze the impact of the various parameters such as the UAV altitude, D2D density, and SINR threshold on the performance evaluation metrics.

In Figures \ref{Pc_D2D} and \ref{Pc_DU}, we show, respectively, the D2D coverage probability and the lower and upper bounds for the DU coverage probability for different SINR detection threshold values. From these figures, we can clearly see that, the analytical and simulation results for D2D match perfectly and the analytical bounds for DU coverage probability and the exact simulation results are close. Figures \ref{Pc_D2D} and \ref{Pc_DU} show that, by increasing the threshold, the coverage probability for D2D users and DUs will decrease. 

\begin{table} [t]
\begin{center}
\captionof{table}{Simulation parameters.}
\resizebox{11cm}{!}{
\begin{tabular}[b]{|c|c|c|} 
\hline
\textbf{Description}  & 	\textbf{Parameter} 	 & \textbf{Value}  \\
\hline
\text{UAV transmit power}   & ${P_u}$  & $5$ W  \\
\hline
\text{D2D transmit power}   & ${P_d}$  & $100$ mW  \\
\hline
\text{Path loss coefficient}   & $K$  & $-30$ dB  \\
\hline
\text{Path loss exponent for UAV-user link}   & ${\alpha _d}$  & $2$  \\
\hline
\text{Path loss exponent for D2D link}   & ${\alpha _u}$  & $3$  \\
\hline
\text{Noise power}   & $N$  & $-120$ dBm  \\
\hline
\text{Bandwidth}   & $W$  & $1$ MHz  \\
\hline
\text{D2D pair fixed distance }   & $d_0$  & $20$ m \\
\hline
\text{Excessive attenuation factor for NLoS }   & $\eta$  & $20$ dB \\
\hline 
\text{Parameters for dense urban environment}   & $B$, $C$  & $0.136$, $11.95$ \\
\hline 
\end{tabular}\vspace{0.4cm} 
}
\end{center}
\end{table}

\begin{figure}[!t]
  \begin{center}
   \vspace{-0.2cm}
    \includegraphics[width=9cm]{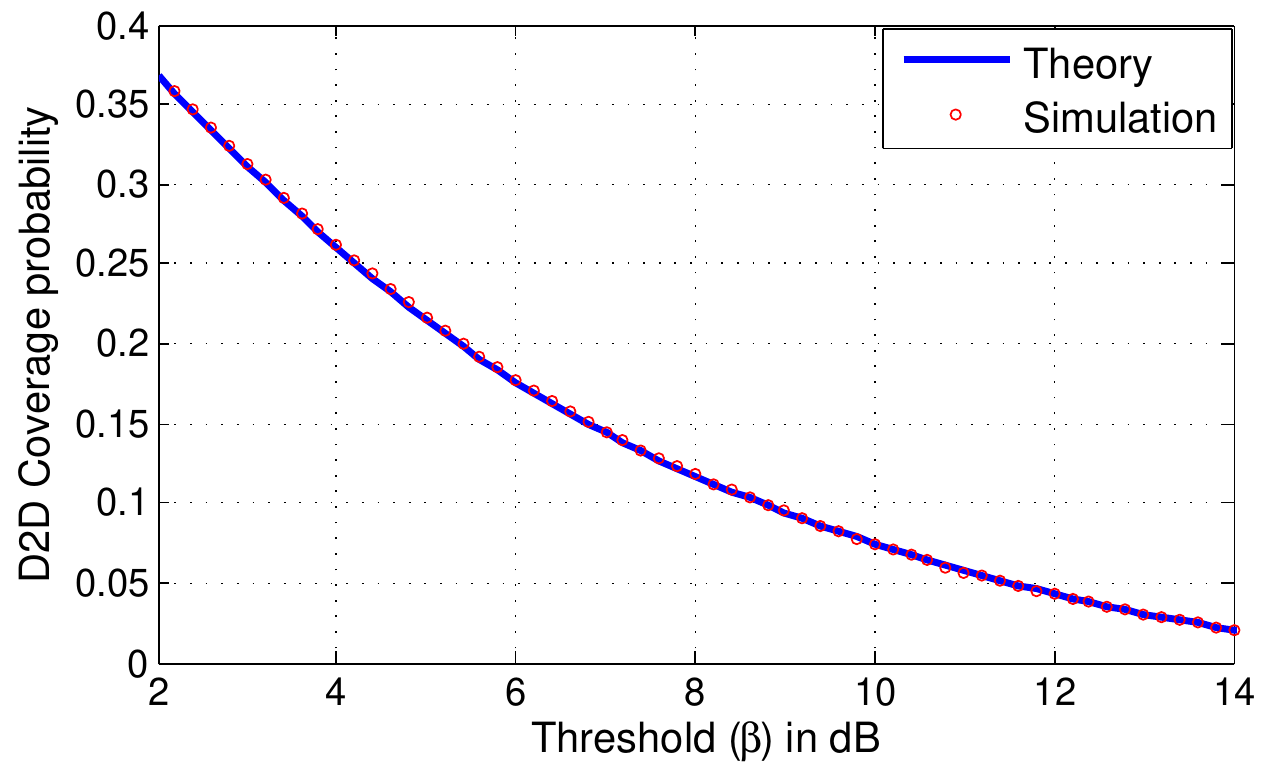}
    \vspace{-0.1cm}
    \caption{{D2D coverage probability vs. SINR threshold}}
    \label{Pc_D2D}
  \end{center}\vspace{-0.4cm}
 \end{figure}

 \begin{figure}[!t]
  \begin{center}
   \vspace{-0.2cm}
    \includegraphics[width=9.5cm]{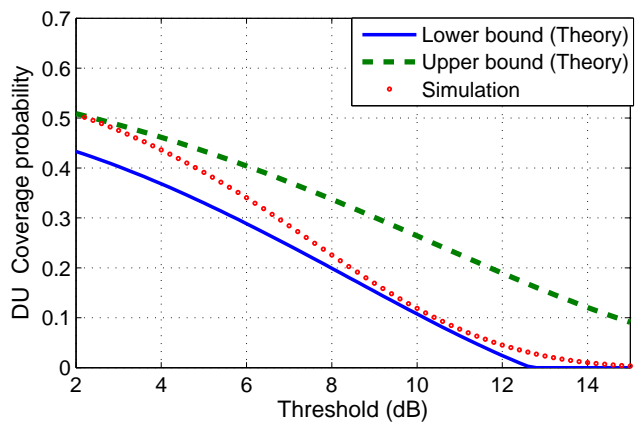}
    \vspace{-0.1cm}
    \caption{{ \textcolor{black}{DU coverage probability vs. SINR threshold. }}}
    \label{Pc_DU}
  \end{center}\vspace{-0.2cm}
\end{figure}


\textcolor{black}{Figure \ref{ASR_beta} illustrates the system sum-rate (Gbps) versus the threshold for 1 MHz transmission bandwidth, ${\lambda _{du}} = {10^{ - 4}}$, $h=500$ m, and two different values of $\lambda_{d}$. By inspecting (23) in Section III, we can see that the rate depends on the coverage probability, which is a decreasing function of the threshold, $\beta$, and an increasing logarithmic function of it. Clearly, for high values of  $\beta$, the received SINR cannot exceed the threshold and, thus, the coverage probabilities tend to zero. On the other hand, according to (20) and (21), as $\beta$ increases, ${\log _2}(1 + \beta )$ increases accordingly. However, since the coverage probability exponentially decreases but ${\log _2}(1 + \beta )$ increases logarithmically, the average rate tends to zero for the high values of $\beta $. Furthermore, for $\beta  \to 0$, since ${\log _2}(1 + \beta )$ tends to zero and the coverage probabilities approach one, the rate becomes zero. }

\begin{figure}[t]
  \begin{center}
   \vspace{-0.2cm}
    \includegraphics[width=9cm]{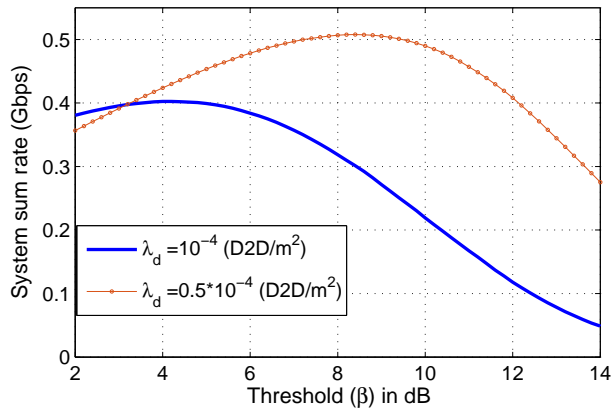}
    \vspace{-0.1cm}
    \caption{ System sum-rate vs. SINR threshold. }
    \label{ASR_beta}
  \end{center}\vspace{-0.2cm}
\end{figure}

Figure \ref{ASR3} shows the impact of D2D density on the sum-rate. In this figure, we can see that a low D2D density yields low interference. However, naturally, decreasing the number of D2D users in an area will also decrease the sum-rate. For high D2D density, high interference reduces the coverage probability and consequently the data rate for each user. However, since the sum-rate is directly proportional to the number of D2D users, increasing the  D2D density can also improve the sum-rate.  According to the Figure 6, as the density of downlink users increases, the optimal ${\lambda _d}$ that maximizes the sum-rate decreases. This is due to the fact that, as ${\lambda _{du}}$  increases, the contribution of DUs in the system sum-rate increases and hence increasing the rate of each DU enhances the system sum-rate. To increase the rate of a DU, the number of D2D users as the interference source for DUs should be reduced. As a result, the optimal ${\lambda _d}$ decreases as as ${\lambda _{du}}$  increases. For instance as shown in the figure, by increasing ${\lambda _{du}}$ from $10^{ - 4}$ to $4 \times {10^{ - 4}}$, the optimal ${\lambda _d}$ decreases from $0.9 \times {10^{ - 4}}$ to $0.3 \times {10^{ - 4}}$.

It is important to note that the value of the fixed distance, $d_0$,  between the D2D pair significantly impacts the rate performance. Figure \ref{UAV_D2D} shows the  ${\bar C_{{\rm{sum}}}}$ as a function of the density of D2D users and $d_0$. From this figure, we can see that, the rate increases as the fixed distance between a D2D receiver and its corresponding transmitter decreases. Moreover, the optimal D2D density which leads to a maximum ${\bar C_{{\rm{sum}}}}$, increases by decreasing $d_0$. In fact, for lower values of $D$ we can have more D2D users in the network. For instance, by reducing $d_0$ from 8 m to 5 m, the optimum average number of D2D users increases by a factor of 3.

\begin{figure}[!t]
  \begin{center}
   \vspace{-0.2cm}
    \includegraphics[width=9cm]{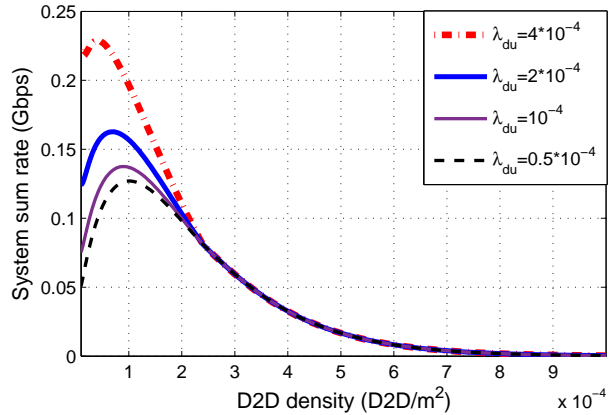}
    \vspace{-0.1cm}
    \caption{ System sum-rate vs. D2D density (number of D2D pairs per $\rm{m}^2$).}
    \label{ASR3}
  \end{center}\vspace{-0.5cm}
\end{figure}

\begin{figure}[!t]
  \begin{center}
   \vspace{-0.2cm}
    \includegraphics[width=9cm]{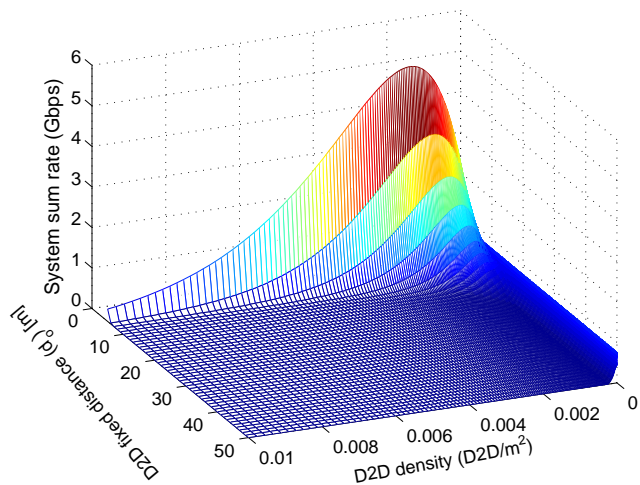}
    \vspace{-0.1cm}
    \caption{ System sum-rate vs. D2D density and $d_0$. }
    \label{UAV_D2D}
  \end{center}\vspace{-0.2cm}
\end{figure}

Figure \ref{Pcov} shows the coverage probability for DUs and D2D users as a function of the UAV altitude. From the DUs' perspective, the UAV should be at an optimal altitude such that it can provide a maximum coverage. In fact, the UAV should not position itself at very low altitudes, due to high shadowing  and a low probability of LoS connections towards the DUs. On the other hand, at very high altitudes, LoS links exist with a high probability but the large distance between UAV and DUs results in a high the path loss. As shown in Figure \ref{Pcov}, for $h= 500$ m the DU coverage probability is maximized. Note that from a D2D user perspective, the UAV creates interference on the D2D receiver. Therefore, D2D users prefer the UAV to be at an altitude for which it provides a minimum coverage radius. As seen in Figure 8, for $h \to \infty$, the D2D users achieve the maximum performance. However, $h= 800$ m results in a minimum D2D coverage probability due the high interference from the UAV. 

\begin{figure}[!t]
  \begin{center}
   \vspace{-0.2cm}
    \includegraphics[width=9cm]{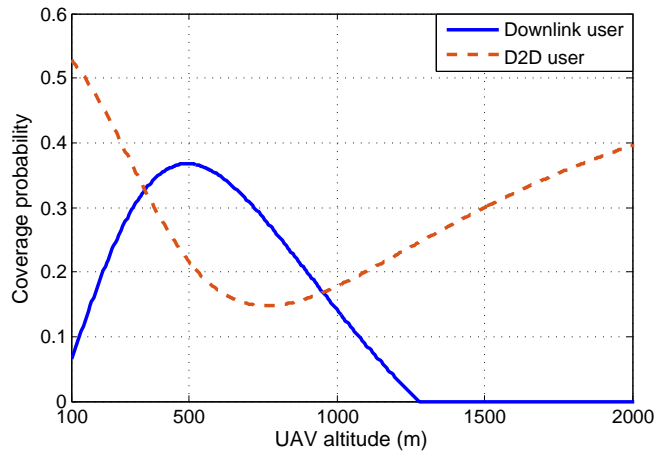}
    \vspace{-0.1cm}
    \caption{Coverage probability vs. UAV altitude.}
    \label{Pcov}
  \end{center}\vspace{-0.2cm}
\end{figure}

 \textcolor{black}{
Figure \ref{Hopt} shows the optimal UAV altitude that maximizes DU coverage probability versus the D2D users' density.  As we can see from Figure \ref{Hopt}, the optimal UAV altitude for downlink users decreases as the number of D2D users increases. This is due to the fact that a higher density of D2D users creates higher interference on the downlink users, and consequently the UAV reduces its altitude to improve SINR value for the downlink users. In other words, the UAV positions itself closer to the downlink users to cope with the high interference caused by the increased number of D2D users. From Figure \ref{Hopt}, we can see that, the optimal UAV altitude is independent of the fixed distance, $d_0$, between the D2D transmitter and receiver pair. In fact, the distance between D2D users does not affect the amount interference generated on the downlink users. Therefore, the optimal altitude of the UAV does not change if $d_0$ changes.  
}  

\begin{figure}[!t]
  \begin{center}
   \vspace{-0.2cm}
    \includegraphics[width=9cm]{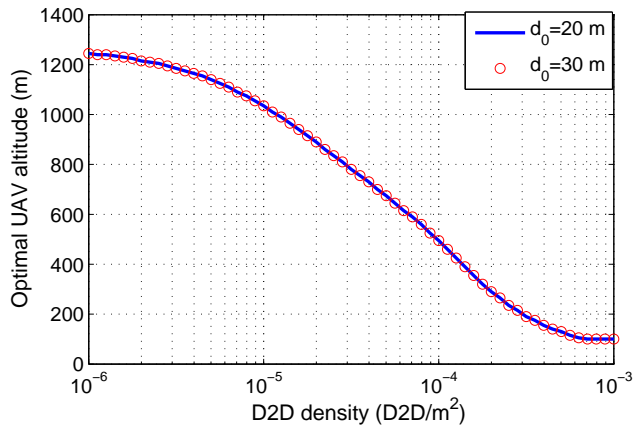}
    \vspace{-0.1cm}
    \caption{ \textcolor{black}{Optimal UAV altitude vs. D2D density.}}
      \label{Hopt}
  \end{center}\vspace{-0.4cm}
\end{figure}

Figure \ref{ASR} shows ${\bar C_{{\rm{sum}}}}$  versus the UAV altitude for different values of the fixed distance, $d_0$, the fixed distance between a D2D transmitter/receiver pair. The optimum values for the height which lead to a maximum ${\bar C_{{\rm{sum}}}}$ are around $300$\,m, $350$\,m, and $400$\,m for $d_0=20\,\rm{m}, 25\,\rm{m}$ and $30$~m. Note that the optimal $h$ that maximizes the sum-rate depends on the density of DU and D2D users. From Figure \ref{ASR}, considering $d_0=20$\,m as an example, we can see that for $h > 1300$\,m, the system sum-rate starts increasing. This stems from the fact that the DU coverage probability tends to zero and, thus,  only D2D users impact ${\bar C_{{\rm{sum}}}}$. Hence, as the UAV moves up in altitude, the interference on D2D users decreases and ${\bar C_{{\rm{d}}}}$  increases. Moreover, for $300\,\textnormal{m} < h < 1300\,\textnormal{m}$, Figure \ref{ASR} shows that the coverage probability and, consequently, the average rate for the downlink users decrease as the altitude increases. However, increasing the UAV altitude reduces the interference on the D2D users and improves the average rate for D2D users. In addition, in this range of $h$, since DUs have more contributions on ${\bar C_{{\rm{sum}}}}$ than the D2D users, ${\bar C_{{\rm{sum}}}}$ is a decreasing function of altitude.

\begin{figure}
  \begin{center}
   \vspace{-0.2cm}
    \includegraphics[width=9cm]{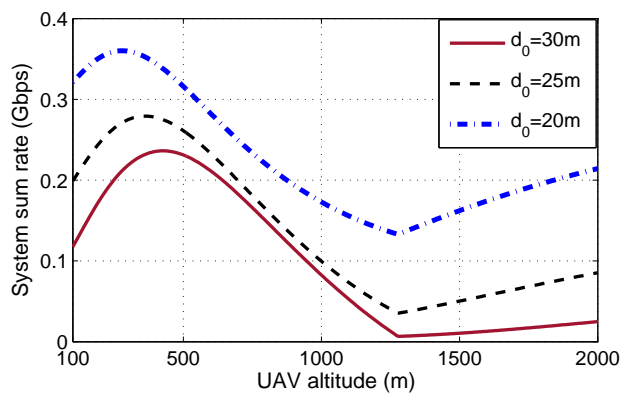}
    \vspace{-0.1cm}
    \caption{ System sum-rate vs. UAV altitude.}
    \label{ASR}
  \end{center}\vspace{-0.4cm}
\end{figure}
\subsection{The mobile UAV scenario}
Next, we study the mobile UAV scenario. In this case, we can satisfy the coverage requirement for all the DUs. In fact, the UAV moves over the target area and attempts to serve the DUs at the stop points to guarantee that all the DUs will be in its coverage radius. 

Figure \ref{Rc_lambda} shows the coverage radius of the mobile UAV when it is located at the optimal altitude as the D2D density varies. As expected, the coverage radius decreases as the D2D density increases. For instance, for $\varepsilon  = 0.6$, when  ${\lambda _d}$  increases from ${10^{ - 5}}$ to ${10^{ - 4}}$, the coverage radius decreases from $1600$\,m to $300$\,m. \textcolor{black}{Moreover, by reducing the minimum coverage requirement of DUs, the UAV can cover a larger area. For instance, reducing $\epsilon$ from 0.6 to 0.4 increases the UAV coverage radius from $290$\,m to $380$\,m for ${\lambda _d=10^{ - 4}}$. Note that, since the main goal of the UAV is to provide coverage for the entire target area, to compensate for the low coverage radius, we should increase the number of stop points for serving the DUs and consequently a longer time is required for the full coverage.}

\begin{figure}[!t]
  \begin{center}
   \vspace{-0.2cm}
    \includegraphics[width=9cm]{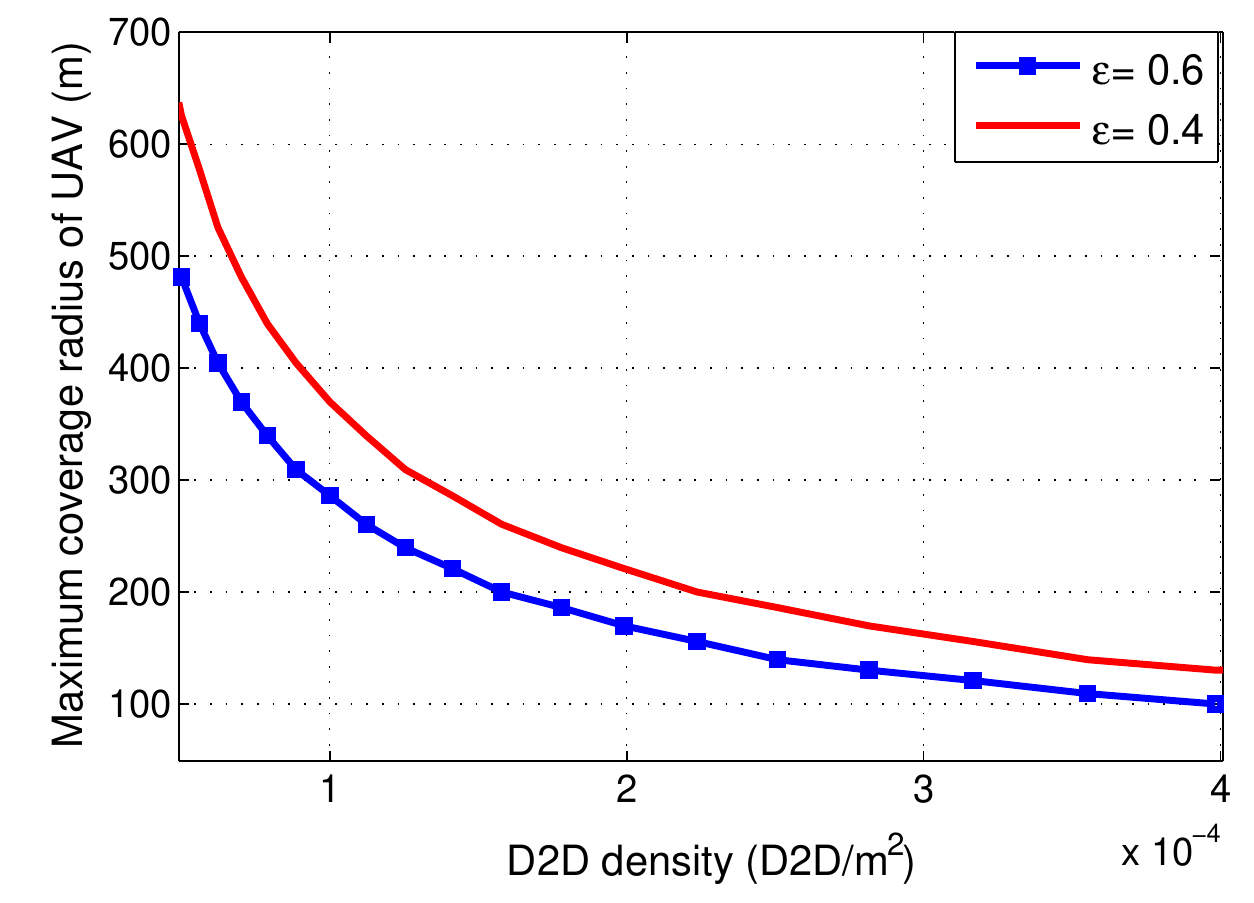}
    \vspace{-0.1cm}
    \caption{\label{fig: Fig1} Maximum UAV coverage radius vs. D2D density (number of D2D pairs per $\rm{m}^2$).}
    \label{Rc_lambda} 
  \end{center}\vspace{-0.2cm} 
\end{figure}

In Figure \ref{Stop_points}, we show the minimum number of stop points as a function of the D2D user density. In this figure, we can see that, as expected, the number of stop points must increase when the density of D2D users increases. In fact, to overcome the higher interference caused by increasing the number of D2D users, the UAV will need more stop points to satisfy the DUs' coverage constraints. For instance, when ${\lambda_d}$ increases from $0.2 \times {10^{ - 4}}$ to $0.8 \times {10^{ - 4}}$, the number of stop points must be increased from 3 to 8. 
Note that, when computing the minimum number of stop points for each ${\lambda_d}$, we considered optimal values for the UAV altitude such that it can provide a maximum coverage for the DUs. Therefore, the UAVs altitude changes according to the D2D density. Moreover, as seen from Figure \ref{Stop_points}, the minimum number of stop points remains constant for a range of  ${\lambda_d}$. This is due to the fact that the number of stop points is an integer and hence, for different values of ${\lambda_d}$, the integer value will be the same. However, although the minimum number of stop points for two different D2D densities  are the same, the UAV can transmit with lower power in the case of lower D2D density. 

In Figure \ref{StopPoints_vs_H}, we show the minimum number of stop points as a function of the UAV altitude for $\lambda_d=10^{-4}$. Figure  \ref{StopPoints_vs_H} shows that, for some values of $h$ which correspond to the optimal UAV altitude, the minimum number of stop points is minimized. For example, the range of optimal $h$ for $\epsilon=0.4$ and $\epsilon=0.6$ is, respectively, $400\,\textnormal{m}< h<500\,\textnormal{m}$ and $300\,\textnormal{m}<h<350\,\textnormal{m}$. As expected, the minimum number of stop points is lower for the lower value of  $\epsilon$.  


\begin{figure}[!t]
  \begin{center}
   \vspace{-0.2cm}
    \includegraphics[width=9cm]{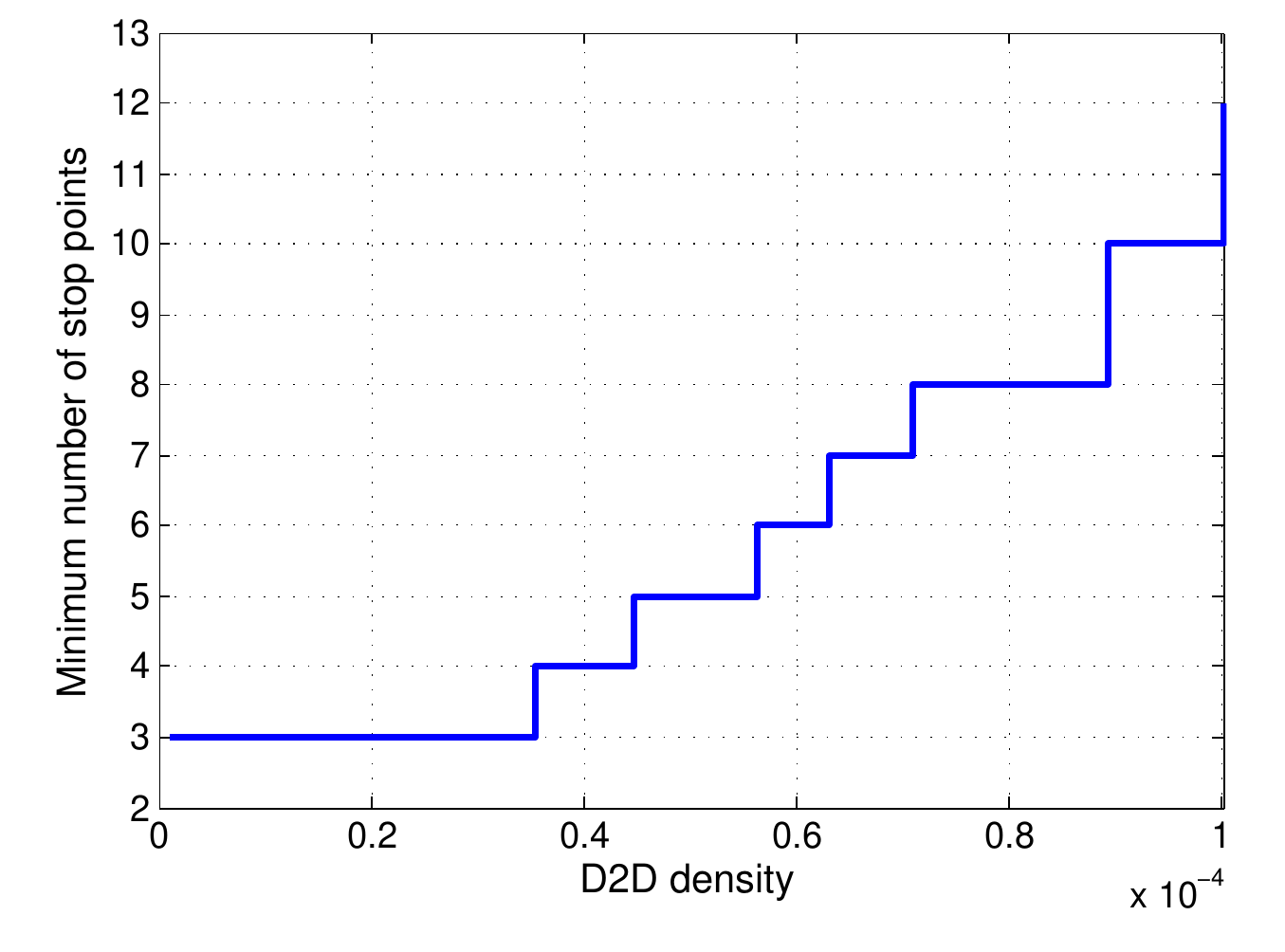}
    \vspace{-0.1cm}
    \caption{ Number of stop points vs.  D2D density.}  
    \label{Stop_points}
  \end{center}\vspace{-0.2cm}
\end{figure}

\begin{figure}[!t]
  \begin{center}
   \vspace{-0.2cm}
    \includegraphics[width=9cm]{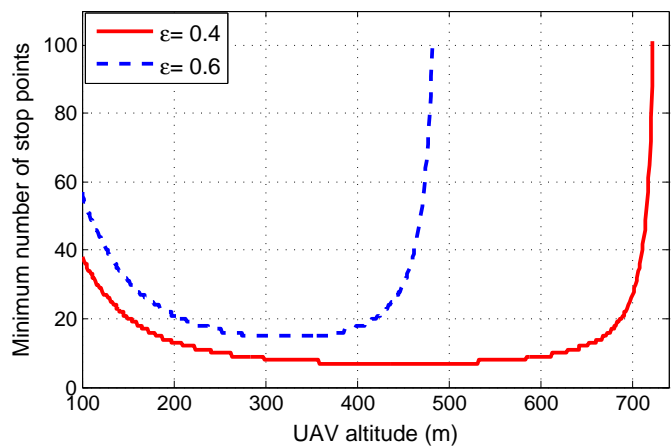}
    \vspace{-0.1cm}
    \caption{ Minimum number of stop points vs. UAV altitude.}
    \label{StopPoints_vs_H}
  \end{center}\vspace{-0.2cm}
\end{figure}

\textcolor{black}{
Figure \ref{Coverage_delay} shows the tradeoff between the downlink coverage probability and the delay which is considered to be proportional to the number of stop points. In Figure \ref{Coverage_delay}, we can see that, in order to guarantee a higher coverage probability for DUs, the UAV should stop at more locations. As observed in this Figure, for ${\lambda_{d}=10^{-4}}$, to increase the DU coverage probability from 0.4 to 0.7, the number of stop points should increase from 5 to 23. For a higher number of stop points, the UAV is closer to the DUs and, thus, it has a higher chance of LoS. However, on the average, a DU should wait for a longer time to be covered by the UAV that reaches its vicinity. In addition, as the density of D2D users increases, the number of stop points (delay) increases especially when a high coverage probability for DUs must be satisfied. For instance, if ${\lambda_{d}}$ increases from $0.5\times10^{-4}$ to $10^{-4}$, or equivalently from $50$ to $100$ for the given area, the number of stop points should increase from 4 to 9 to satisfy a 0.5 DU coverage probability, and from 20 to 55 for a 0.8 coverage requirement.
}
\begin{figure}[!t]
  \begin{center}
   \vspace{-0.2cm}
    \includegraphics[width=9cm]{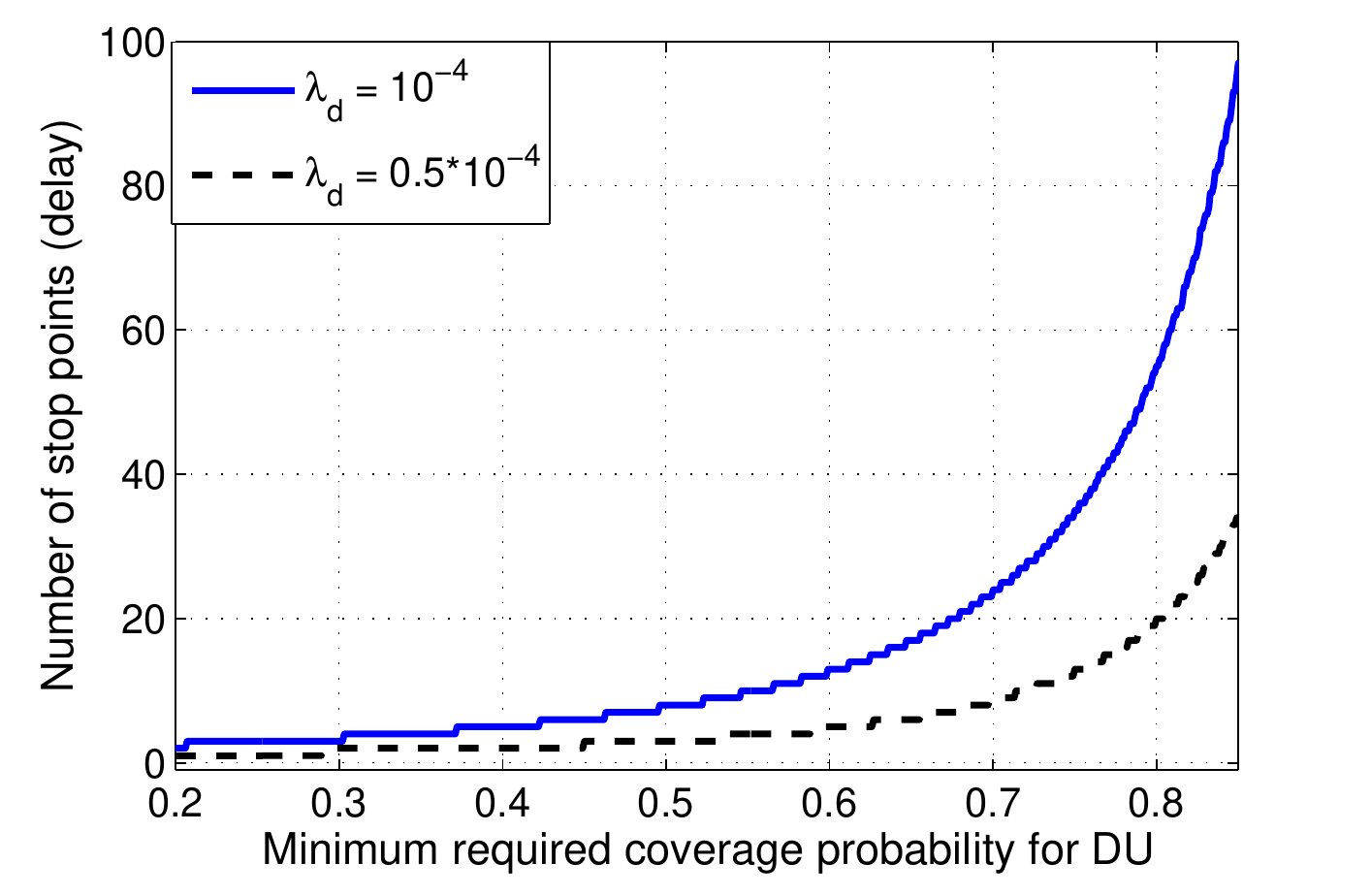}
    \vspace{-0.1cm}
    \caption{Minimum number of stop points vs. coverage probability (coverage-delay tradeoff)}
    \label{Coverage_delay} 
  \end{center}\vspace{-0.2cm}
\end{figure}

 \textcolor{black}{
Figure \ref{OutageProbability} shows the overall outage probability for D2D users versus the number of retransmissions. As the number of retransmissions (time slots) increases, the overall outage probability also increases. In other words, for higher number of time slots, the possibility that a failure happens during retransmissions, increases. Furthermore, since the UAV is an interference source for the D2D users, the higher number of stop points leads to a higher outage probability. From Figure \ref{OutageProbability}, we can see that, the increase in the outage probability of D2D users due to the UAV is 0.20 for $M=3$, and is 0.38 for $M=7$. Therefore, when the number of stop points increases due to the higher density of D2D users or a higher coverage requirement of the downlink users, the D2D communications are more prone to a failure.  
}

\begin{figure}[!t]
  \begin{center}
   \vspace{-0.2cm}
    \includegraphics[width=9cm]{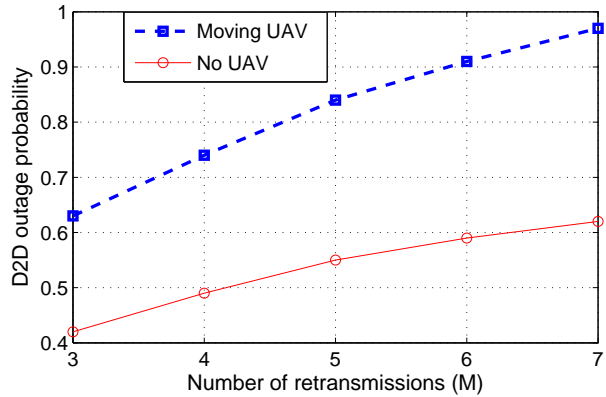}
    \vspace{-0.1cm}
    \caption{ \textcolor{black}{Overall D2D outage probability vs. number of retransmissions.}}  
    \label{OutageProbability}
  \end{center}\vspace{-0.2cm}
\end{figure}

\section{Conclusions}
\textcolor{black}{In this paper, we have studied the performance of a UAV that acts as a flying base station in an area in which users are engaged in D2D communication. We have considered two types of users: in the network: the downlink users served by the UAV and D2D users that communicate directly with one another. For both types, we have derived tractable expressions for the coverage probabilities as the main performance evaluation metrics. The results have shown that a maximum system sum-rate can be achieved if the UAV altitude is appropriately adjusted based on the D2D users density. In the mobile UAV scenario, using the disk covering problem, the entire target area (cell) can be completely covered by the UAV in a shortest time with a minimum required transmit power. Moreover, in this case, we have derived the overall outage probability for D2D users, and showed that the outage probability increases as the number of stop point increases. Finally, we have analyzed the tradeoff between the coverage and the time required for covering the entire target area (delay) by the mobile UAV. The results show that, the number of stop points must be significantly increased as the minimum coverage requirement for DUs increases.\vspace{-0.6cm}
}
\appendix
\subsection{Proof of Theorem 1}\vspace{-0.3cm}
\begin{align}
{P_{{\mathop{\rm cov}} ,{{d}}}}(r,\varphi ,\beta ) &={\mathds {P}}\left[ {\gamma _d} \ge \beta \right]= \mathds {P}\left[ \frac{{{P_d}{d_0^{ - {\alpha _d}}}g}}{{I_d^c + {I_u} + N}} \ge \beta \right]\nonumber\\
&= \mathds {P}\left[ g \ge \frac{{\beta {d_0^{{\alpha _d}}}(I_d^c + {I_u} + N)}}{{{P_d}}}\right]\mathop  = \limits^{(a)}{\mathop{\mathds E_{{I_u},I_d^c}}}\left[ {\exp (\frac{{ - \beta {d_0^{{\alpha _d}}}(I_d^c + {I_u} + N)}}{{{P_d}}})} \right] \nonumber\\
&\mathop  = \limits^{(b)} {{\mathds {E}}_{{I_u}}}\left[ {\exp (\frac{{ - \beta {d_0^{{\alpha _d}}}{I_u}}}{{{P_d}}})} \right]{{\mathds {E}}_{I_d^c}}\left[ {\exp (\frac{{ - \beta {d_0^{{\alpha _d}}}I_d^c}}{{{P_d}}})} \right]\exp \left( {\frac{{ - \beta {d_0^{{\alpha _d}}}N}}{{{P_d}}}} \right),
\end{align}
where $g$ is an exponential random variable with a mean value of one (i.e.  $g\sim\text{exp}(1)$), (\textit a) follows from the exponential distribution of $g$ based on the Rayleigh fading assumption, and taking the expectation over ${I_u}$ and  ${I_d^c}$ (as random variables). Step (\textit b) comes from the fact that ${I_u}$ and ${I_d^c}$ are independent because the interference stems from different sources which are spatially uncorrelated.

Here, ${{\mathds {E}}_{{I_u}}}$ and ${{\mathds {E}}_{I_d^c}}$ are given by:

\begin{align} 
{{\mathds {E}}_{{I_u}}}\left[ {\exp (\frac{{ - \beta {d_0^{{\alpha _d}}}{I_u}}}{{{P_d}}})} \right] &= {{{P}}_{{\rm{LoS}}}}\exp \left( {\frac{{ - \beta {d_0^{{\alpha _d}}}{P_u}{{|X_u|}^{ - {\alpha _u}}}}}{{{P_d}}}} \right)\nonumber\\
&  + {{{P}}_{{\rm{NLoS}}}}\exp \left( {\frac{{ - \beta {d_0^{{\alpha _d}}}\eta {P_u}{|X_u|^{ - {\alpha _u}}}}}{{{P_d}}}} \right),
\end{align} 
 
\begin{align}
{{\mathds {E}}_{I_d^c}}\left[ {\exp (\frac{{ - \beta {d_0^{{\alpha _d}}}I_d^c}}{{{P_d}}})} \right]& = {{\mathds {E}}_{{d_{i}},{g_{i}}}}\left[ {\prod\limits_{i} {\exp (\frac{{ - \beta {d_0^{{\alpha _d}}}}}{{{P_d}}}{P_d}{d_{i}}^{ - {\alpha _d}}{g_{i}})} } \right]\mathop  = \limits^{(a)}  \exp \left( {\frac{{ - 2{\pi ^2}{\lambda _d}{\beta ^{2/{\alpha _d}}}{d_0^2}}}{{{\alpha _d}\sin (2\pi /{\alpha _d})}}} \right),
\end{align}
 \textcolor{black}{where the details of step (a) follow directly from the results in \cite{martin}.}

Finally, using (31), (32) and (33) Theorem 1 is proved. 
\subsection{Proof of Theorem 2}
The coverage probability for a cellular user located at $(r,\varphi )$ is written as: \vspace{0.1cm}
\begin{align}
{P_{{\mathop{\rm cov}} ,{du}}}(r,\varphi ,\beta ) = \mathds {P}\left[ {\gamma _u} \ge \beta \right]= {{{P}}_{{\rm{LoS}}}}(r)\mathds {P}\left[ \frac{{{P_u}{r^{ - {\alpha _u}}}}}{{{I_d} + N}} \ge \beta \right]  + {{{P}}_{{\rm{NLoS}}}}(r)\mathds {P}\left[ \frac{{\eta {P_u}{r^{ - {\alpha _u}}}}}{{{I_d} + N}} \ge \beta \right] \nonumber\\
 = {{{P}}_{{\rm{LoS}}}}(r)\mathds {P}\left[ {I_d} \le \frac{{{P_u}{r^{ - {\alpha _u}}} - \beta N}}{\beta }\right]  + {{{P}}_{{\rm{NLoS}}}}(r) \mathds {P}\left[ {I_d} \le \frac{{\eta {P_u}{r^{ - {\alpha _u}}} - \beta N}}{\beta }\right].
\end{align}

Note that, there is no closed-form expression for the  cumulative distribution function (CDF) of the interference from D2D users \cite{ganti} and \cite{weber}. Here, we provide lower and upper bounds for the CDF of interference.
First, we divide the interfering D2D transmitters into two subsets:
\begin{equation}
\left\{ \begin{array}{l}
{\Phi _1} = \{ {\Phi _{\rm{B}}}|{P_d}{d_{i}}^{ - {\alpha _d}}{g_{i}} \ge T\}, \\
{\Phi _2} = \{ {\Phi _{\rm{B}}}|{P_d}{d_{i}}^{ - {\alpha _d}}{g_{i}} \le T\},
\end{array} \right.
\end{equation}
where $T$ is a threshold which is used to derive the CDF of the interference from D2D users.\\
Now, considering the interference power from D2D users located in ${\Phi _1}$  and ${\Phi _2}$  as ${I_{d,{\Phi _1}}}$ and ${I_{d,{\Phi _2}}}$, we have:
\begin{align}
\mathds {P}\left[ {I_d} \le T\right]  &= \mathds {P}\left[ {I_{d,{\Phi _1}}} + {I_{d,{\Phi _2}}} \le T\right]  \le \mathds {P}\left[ {I_{d,{\Phi _1}}} \le T\right]  = \mathds {P}\left[ {\Phi _1} = 0\right] \nonumber\\
 &= {\mathds {E}}\left[ {\prod\limits_{{\Phi _B}} {\mathds {P}({P_d}{d_{i}}^{ - {\alpha _d}}{g_{i}} < T)} } \right] = {\mathds {E}}\left[ {\prod\limits_{{\Phi _B}} {\mathds {P}({g_{i}} < \frac{{T{d_{i}}^{{\alpha _d}}}}{{{P_d}}})} } \right]\nonumber\\
& \mathop  = \limits^{(a)} {\mathds {P}}\left[ {\prod\limits_{{\Phi _B}} {1 - \text{exp}( - \frac{{T{d_{i}}^{{\alpha _d}}}}{{{P_d}}})} } \right] \mathop  = \limits^{(b)} \exp\left( { - \lambda _d \int\limits_0^\infty  {\text{exp}( - \frac{{T{r^{{\alpha _d}}}}}{{{P_d}}}){\rm{ }}r\text{d}r} } \right)\nonumber\\
& = \exp \left( { - \pi \lambda _d {{\left( {\frac{T}{{{P_d}}}} \right)}^{ - 2/{\alpha _d}}}\Gamma (1 + 2/{\alpha _d})} \right),
\end{align}
where ($a$) and ($b$) come from the Rayleigh fading assumption and PGFL of the PPP.

The upper bound is derived as follows:
\begin{align}
\mathds {P}\left[ {I_d} \le T\right]  &= 1 - \mathds {P}\left[ {I_d} \ge T\right]  \nonumber\\
&=1 - \Big( {\mathds {P}\left[ {I_d} \ge T|{I_{d,{\Phi _1}}} \ge T\right] \mathds {P}\left[ {I_{d,{\Phi _1}}} \ge T\right]  + \mathds {P}\left[ {I_d} \ge T|{I_{d,{\Phi _1}}} \le T\right] \mathds {P}\left[ {I_{d,{\Phi _1}}} \le T\right] } \Big)\nonumber\\
& = 1 -\Big( {\mathds {P}\left[ {I_{d,{\Phi _1}}} \ge T\right]  + \mathds {P}\left[ {I_d} \ge T|{I_{d,{\Phi _1}}} \le T\right]  \mathds {P}\left[ {I_{d,{\Phi _1}}} \le T\right] } \Big)\nonumber\\
& = 1 - \Big( {1 - \mathds {P}\left[ {\Phi _1} = 0\right]  + \mathds {P}\left[ {I_d} \ge T|{I_{d,{\Phi _1}}} \le T\right] \mathds {P}\left[ {\Phi _1} = 0\right] } \Big)\nonumber\\
& = \mathds {P}\left[ {\Phi _1} = 0\right] \Big( {1 - \mathds {P}\left[ {I_d} \ge T|{\Phi _1} = 0\right]} \Big).
 \end{align}
 
Also,
\begin{align}
\mathds {P}\left[ {I_d} \ge T|{\Phi _1} = 0\right]  &\mathop  \le \limits^{(a)} \frac{{{\mathds {E}}\left[ {I_d} \ge T|{\Phi _1} = 0\right] }}{T}= \frac{1}{T}{\mathds {E}}\left[ {\sum\limits_\Phi  {{P_d}{d_{i}}^{ - {\alpha _d}}{g_{i}}\mathds{1}{{({P_d}{d_{i}}^{ - {\alpha _d}}{g_{i}} \le T)}}} } \right]\nonumber\\
 &= \frac{1}{T}{{\mathds {E}}_{{d_{i}}}}\left[ {\sum\limits_\Phi  {{P_d}{d_{i}}^{ - {\alpha _d}}{{\mathds {E}}_{{g_{i}}}}\left[ {{g_{i}}\mathds{1}({g_{i}} \le \frac{{T{d_{i}}^{{\alpha _d}}}}{{{P_d}}})} \right]} } \right]\nonumber\\
&= \frac{1}{T}{{\mathds {E}}_{{d_{i}}}}\left[ {\sum\limits_\Phi  {{P_d}{d_{i}}^{ - {\alpha _d}}\left[ {\int\limits_0^{{\textstyle{{T{d_{i}}^{{\alpha _d}}} \over {{P_d}}}}} {g{e^{ - g}}\text{d}g} } \right]} } \right] \nonumber\\
 & =\frac{{2\pi {P_d}\lambda_d }}{T}\int\limits_0^\infty  {{r^{ - {\alpha _d}}}} \left( {\int\limits_0^{{\textstyle{{T{r^{{\alpha _d}}}} \over {{P_d}}}}} {g{e^{ - g}}\text{d}g{\rm{ }}} } \right)r\text{d}r\nonumber\\
 &= \frac{{2\pi \lambda_d {\rm{ }}\Gamma (1 + 2/{\alpha _d})}}{{{\alpha _d} - 2}}{\left( {\frac{T}{{{P_d}}}} \right)^{ - 2/{\alpha _d}}},
\end{align}
where ($a$) is based on the Markov's inequality which is stated as follows: for any non-negative integrable random variable $X$ and positive $L$,  $P(X \ge L) \le \frac{{\mathds {E}\left[ X \right]}}{L}$.
Also, $\mathds{1}(.)$ is the indicator function which can only be equal to 1 or 0.
Hence, the lower ($L_I$) and upper (${U_I}$) bounds for the CDF of interference  become:
\begin{equation}
\begin{small}
{L_I}(T) = \left( {1 - \frac{{2\pi { \lambda_d} {\rm{ }}\Gamma (1 + 2/{\alpha _d})}}{{{\alpha _d} - 2}}{{\left( {\frac{T}{{{P_d}}}} \right)}^{ - 2/{\alpha _d}}}} \right)\exp\left( { - \pi{ \lambda_d} {{\left( {\frac{T}{{{P_d}}}} \right)}^{ - 2/{\alpha _d}}}\Gamma (1 + 2/{\alpha _d})} \right),
\end{small}
\end{equation}

\begin{equation}
{U_I}(T) =  \exp\left( { - \pi { \lambda_d} {{\left( {\frac{T}{{{P_d}}}} \right)}^{ - 2/{\alpha _d}}}\Gamma (1 + 2/{\alpha _d})} \right).
\end{equation}\\
Thus, we have ${L_I}(T) \le\mathds {P}\{ {I_d} \le T\}  \le {U_I}(T)$.

 \textcolor{black}{
Finally, considering (34), (39), and (40), the lower bound and upper bound of the average coverage probability for DUs in the cell is expressed as:
\begin{align}
{\bar P^L_{{\mathop{\rm cov}} ,\textnormal{du}}}(\beta ) & = \int\limits_0^{{R_c}} {{{{P}}_{{\rm{LoS}}}}(r)} {L_I}\left(\frac{{{P_u}{|X_u|}^{ - {\alpha _u}}}}{\beta } - N\right)\frac{{2r}}{{R_c^2}}{\rm{d}}r \nonumber\\ &+\int\limits_0^{{R_c}} {{{{P}}_{{\rm{NLoS}}}}(r)} {L_I}\left(\frac{{\eta {P_u}{|X_u|}^{ - {\alpha _u}}}}{\beta } - N\right)\frac{{2r}}{{R_c^2}}{\rm{d}}r,\\
{\bar P^U_{{\mathop{\rm cov}} ,\textnormal{du}}}(\beta ) &= \int\limits_0^{{R_c}} {{{{P}}_{{\rm{LoS}}}}(r)} {U_I}\left(\frac{{{P_u}{|X_u|}^{ - {\alpha _u}}}}{\beta } - N\right)\frac{{2r}}{{R_c^2}}{\rm{d}}r \nonumber\\ &+\int\limits_0^{{R_c}} {{{{P}}_{{\rm{NLoS}}}}(r)} {U_I}\left(\frac{{\eta {P_u}{|X_u|}^{ - {\alpha _u}}}}{\beta } - N\right)\frac{{2r}}{{R_c^2}}{\rm{d}}r,
\end{align}
and Theorem 2 is proved. 
}
 \textcolor{black}{
\subsection{Proof of Theorem 3}
Consider $\gamma_{d,i}$ and ${g_i}$, respectively, the SINR and the channel gain (with exponential distribution) at $i^{th}$ retransmission, for $ 1\le i\le M$. The outage probability is the probability of having at least one failure during $M$ retransmissions. Then, we have:
\begin{align}
{P_{out,d}} =& 1 - \mathds {P}\left[ {{\gamma _{d,1}} \ge \beta ,...,{\gamma _{d,M}} \ge \beta } \right] \nonumber  \\
=& 1 - \mathds {P}\left[ {\frac{{{P_d}d_0^{ - {\alpha _d}}{g_1}}}{{I_{d,1}^c + {I_{u,1}} + N}} \ge \beta ,...,\frac{{{P_d}d_0^{ - {\alpha _d}}{g_M}}}{{I_{d,M}^c + {I_{u,M}} + N}} \ge \beta } \right] \nonumber \\
=&1 - \mathds {P}\left[ {{g_1} \ge \frac{{d_0^{{\alpha _d}}\beta (I_{d,1}^c + {I_{u,1}} + N)}}{{{P_d}}},...,{g_M} \ge \frac{{d_0^{{\alpha _d}}\beta (I_{d,M}^c + {I_{u,M}} + N)}}{{{P_d}}}} \right]  \nonumber \\
\mathop  = \limits^{(a)}&1 - \mathds {E}\left[ {\prod\limits_{i = 1}^M {\exp \left( {\frac{{ - d_0^{{\alpha _d}}\beta (I_{d,i}^c + {I_{u,i}} + N)}}{{{P_d}}}} \right)} } \right] \nonumber \\
\mathop  = \limits^{(b)}& 1 - \mathds {E}\left[ {\prod\limits_{i = 1}^M {\exp \left( {\frac{{ - d_0^{{\alpha _d}}\beta I_{d,i}^c}}{{{P_d}}}} \right)} } \right]\mathds {E}\left[ {\prod\limits_{i = 1}^M {\exp \left( {\frac{{ - d_0^{{\alpha _d}}\beta {I_{u,i}}}}{{{P_d}}}} \right)}} \right]\exp \left( {\frac{{ - d_0^{{\alpha _d}}\beta MN}}{{{P_d}}}} \right),
\end{align}
where $(a)$ follows the assumption that the fading is independent in different retransmissions, and step $(b)$ comes from the fact that interference due to D2D users, interference from UAV, and noise are all independent.}  
 \textcolor{black}{
Also, 
\begin{align}
\mathds {E}\left[ {\prod\limits_{i = 1}^M {\exp \left( {\frac{{ - d_0^{{\alpha _d}}\beta I_{d,i}^c}}{{{P_d}}}} \right)} } \right] &= \mathds {E}\left[ {\exp \left( {\frac{{ - d_0^{{\alpha _d}}\beta \sum\limits_{i = 1}^M {I_{d,i}^c} }}{{{P_d}}}} \right)} \right]\nonumber \\
&\mathop  = \limits^{(c)} \exp \left( { - {\lambda _d}\int_{{R^2}} {\left[ {1 - {{\left( {\frac{1}{{1 + \frac{{\beta |x{|^{ - {\alpha _d}}}}}{{d_0^{- {\alpha _d}}}}}}} \right)}^M}} \right]{\rm{d}}x} } \right),
\end{align}
where details of $(c)$ can be found in \cite{martin} where the correlation between D2D interference in different retransmissions is taken into account. Finally, }

\begin{small}
\begin{align}
\prod\limits_{i = 1}^M {{\mathds {E}_{{I_{u,i}}}}\left[ {\exp \left( {\frac{{ - d_0^{{\alpha _d}}\beta {I_{u,i}}}}{{{P_d}}}} \right)} \right]}\mathop  = \limits^{(d)} \prod\limits_{i = 1}^M {\left[ {{P_{{\rm{LoS,}}i}}\exp \left( {\frac{{ - \beta d_0^{{\alpha _d}}{P_u}|{X_{u,i}}{|^{ - {\alpha _d}}}}}{{{P_d}}}} \right) + {P_{{\rm{NLoS,}}i}}\exp \left( {\frac{{ - \beta d_0^{{\alpha _d}}\eta {P_u}|{X_{u,i}}{|^{ - {\alpha _d}}}}}{{{P_d}}}} \right)} \right]}, 
\end{align}
\end{small}where step $(d)$ is based on the fact that the interference from the UAV can be treated as independent in different retransmissions.

Finally, using (40), (41), and (42), Theorem 3 is proved.\vspace{-0.4cm}

\def\baselinestretch{1.5}
\bibliographystyle{IEEEtran}
\bibliography{references}

\end{document}